\documentclass[ejs]{imsart}

\doi{10.1214/154957804100000000}
\pubyear{0000}
\volume{0}
\firstpage{0}
\lastpage{0}

\usepackage{natbib}
\usepackage{graphicx}
\usepackage{amsmath, amsthm, amssymb, amsfonts}
\usepackage{url}
\usepackage{multirow}
\usepackage{rotating}
\usepackage{subcaption}
\usepackage{color} 
\usepackage[colorlinks,citecolor=blue,urlcolor=blue]{hyperref}

\startlocaldefs

\theoremstyle{plain}
\newtheorem{thm}{Theorem}[section]
\theoremstyle{plain}

\theoremstyle{plain}
\newtheorem{lem}{Lemma}[section]
\theoremstyle{plain}

\theoremstyle{plain}
\newtheorem{rem}{Remark}[section]
\theoremstyle{plain}

\newcommand{\E}{\mathbb{E}}
\newcommand{\pr}{\mathbb{P}}

\newcommand{\sgn}{\mbox{sign}}
\newcommand{\var}{\mbox{var}}

\newcommand{\iid}{\mbox{\scriptsize{i.i.d}}}
\newcommand{\bbI}{\mathbb{I}}
\newcommand{\R}{\mathbb{R}}

\newcommand{\vep}{\varepsilon}
\newcommand{\bvep}{\boldsymbol{\varepsilon}}
\newcommand{\ept}{\epsilon_T}
\newcommand{\heta}{\wh{\eta}}
\newcommand{\bpsi}{\boldsymbol{\psi}}
\newcommand{\bdel}{\boldsymbol{\delta}}
\newcommand{\bchi}{\boldsymbol{\chi}}
\newcommand{\bxi}{\boldsymbol{\xi}}

\newcommand{\bb}{\mathbf{b}}

\newcommand{\bE}{\mathbf{E}}

\newcommand{\bg}{\mathbf{g}}
\newcommand{\bp}{\mathbf{p}}

\newcommand{\bu}{\mathbf{u}}

\newcommand{\bx}{\mathbf{x}}

\newcommand{\cC}{\mathcal{C}}
\newcommand{\cD}{\mathcal{D}}
\newcommand{\cE}{\mathcal{E}}

\newcommand{\cI}{\mathcal{I}}

\newcommand{\cN}{\mathcal{N}}

\newcommand{\cR}{\mathcal{R}}
\newcommand{\cT}{\mathcal{T}}
\newcommand{\cU}{\mathcal{U}}
\newcommand{\cX}{\mathcal{X}}

\newcommand{\cTj}{\mathcal{T}^{\mbox{\tiny{Jirak}}}}
\newcommand{\cTh}{\mathcal{T}^{\mbox{\tiny{HH}}}}
\newcommand{\cTlin}{\mathcal{T}^{\mbox{\tiny{linear}}}}
\newcommand{\cTscan}{\mathcal{T}^{\mbox{\tiny{scan}}}}
\newcommand{\cTeh}{\mathcal{T}^{\mbox{\tiny{EH}}}}
\newcommand{\cTs}{\mathcal{T}^{\mbox{\tiny{SBS}}}}
\newcommand{\uDel}{\underline{\Delta}}
\newcommand{\cDzh}{\wt{\cD}}
\newcommand{\cTzh}{\wt{\cT}}
\newcommand{\pizh}{\wt{\pi}}

\newcommand{\inner}[2]{\langle #1, #2 \rangle}
\def\wh{\widehat}
\def\wt{\widetilde}
\newcommand{\beq}{\begin{eqnarray}}
\newcommand{\eeq}{\end{eqnarray}}
\newcommand{\beqs}{\begin{eqnarray*}}
\newcommand{\eeqs}{\end{eqnarray*}}
\endlocaldefs

\begin{document}

\begin{frontmatter}

\title{Change-point detection in panel data via double CUSUM statistic}
\runtitle{Change-point detection in panel data}


\author{\fnms{Haeran} \snm{Cho}\corref{}\ead[label=e1]{haeran.cho@bristol.ac.uk}}
\address{School of Mathematics, University of Bristol, UK.}

\runauthor{H. Cho}

\begin{abstract}
In this paper, we consider the problem of (multiple) change-point detection in panel data. We propose the double CUSUM statistic which utilises the cross-sectional change-point structure by examining the cumulative sums of ordered CUSUMs at each point. The efficiency of the proposed change-point test is studied, which is reflected on the rate at which the cross-sectional size of a change is permitted to converge to zero while it is still detectable. Also, the consistency of the proposed change-point detection procedure based on the binary segmentation algorithm, is established in terms of both the total number and locations (in time) of the estimated change-points. Motivated by the representation properties of the Generalised Dynamic Factor Model, we propose a bootstrap procedure for test criterion selection, which accounts for both cross-sectional and within-series correlations in high-dimensional data. The empirical performance of the double CUSUM statistics, equipped with the proposed bootstrap scheme, is investigated in a comparative simulation study with the state-of-the-art.
As an application, we analyse the log returns of S\&P 100 component stock prices over a period of one year.
\end{abstract}



\received{\smonth{11} \syear{2015}}


\end{frontmatter}

\section{Introduction}
\label{sec:intro}

Multivariate, possibly high-dimensional observations over time have emerged
in many fields, such as economics, finance, natural science, engineering and humanities,
thanks to the advances of computing technologies \citep{fan2011}.
Multivariate data observed in practical problems often appear nonstationary in the sense that 
it is natural to let some quantities or parameters involved in the model to be time-varying.
Arguably, the simplest departure from assuming stationarity
is to operate under the assumption of piecewise stationarity,
which allows more flexibility as well as providing interesting insights into the data 
with regards to the structural change-points.
Besides, in the case of time series analysis, it enables (short-term) prediction of the future process values,
by treating the last estimated segment as being stationary.
Throughout the paper, the term ``multiple change-point detection'' 
is used interchangeably with ``segmentation''.


Panel data models are frequently adopted to analyse high-dimensional data involving measurements over time.
In this paper, we focus on the problem of detecting (possibly) multiple change-points
in the mean of panel data, where $n$, the dimensionality of the data, may increase with the number of observations $T$.
The panel model is presented as
\begin{eqnarray*}
x_{j, t} = f_{j, t} + \vep_{j, t}, \qquad t=1, \ldots, T; \ j=1, \ldots, n,
\end{eqnarray*}
where $\{f_{j, t}\}_{t=1}^T, \, j=1, \ldots, n$ are piecewise constant signals
which share an unknown number of change-points at unknown locations.

CUSUM statistics have been widely adopted for segmenting both univariate and multivariate data.
For univariate data segmentation, CUSUM statistics are computed over time,
and this series of CUSUMs is examined to locate a change-point, 
often as where its maximum in the absolute value is attained.
Combined with a binary segmentation (BS) algorithm,
the CUSUM statistics can consistently detect multiple change-points in a recursive manner
(see e.g., \cite{vostrikova1981}, \cite{venkatraman1992} and \cite{cho2012}).

For segmenting $n$-dimensional panel data, 
we may apply the above procedure to each univariate component series separately,
and then prune down the estimated change-points by identifying those 
detected for the identical change-point across the panel.
However, such pruning may be difficult to accomplish
even in moderately large dimensions, due to the estimation bias present in each change-point estimate.
Besides, this approach does not take into account, and thus benefit from, the cross-sectional nature of change-points
(that they are shared across the panel)
which may lead to loss of power in change-point detection.
Instead, we propose to segment the $n$-dimensional data simultaneously
by searching for change-points from the {\em aggregation} of $n$ series of CUSUM statistics,
rather than from individual CUSUM series separately.

\subsection{Literature review}
\label{sec:lit}

Let $\cC_b$ denote a CUSUM operator which
takes $x_{j, t}$ over a generic interval $t \in [s, e]$ with $1 \le s < e \le T$ as an input and returns
\beq
\cX^j_{s, b, e} &=& \cC_b(\{\sigma_j^{-1} x_{j, t}\}_{t=s}^e) \nonumber
\\
&=& \sqrt{\frac{e-b}{(e-s+1)(b-s+1)}}\sum_{t=s}^b \frac{x_{j, t}}{\sigma_j} -
\sqrt{\frac{b-s+1}{(e-s+1)(e-b)}}\sum_{t=b+1}^e \frac{x_{j, t}}{\sigma_j}
\nonumber
\\
&=& \frac{1}{\sigma_j}\sqrt{\frac{(b-s+1)(e-b)}{e-s+1}}\left( \frac{1}{b-s+1}\sum_{t=s}^b x_{j, t} -
\frac{1}{e-b}\sum_{t=b+1}^e x_{j, t} \right)
\label{eq:cusum:two}
\eeq
for $b=s, \ldots, e-1$, with a suitably chosen scaling constant $\sigma_j$.

Assuming the presence of at most one change-point, 
some change-point tests for panel data have been proposed, 
based on the principle of high-dimensional CUSUM series aggregation.
Note that for single change-point detection, $s=1$ and $e=T$.
\cite{zhang2010} considered a change-point test with the test statistic
$\cT^{\mbox{\tiny{ZSJL}}}_{1, T} = \max_{b \in [1, T)}\sum_{j=1}^n(\cX^j_{1, b, T})^2$
and the test criterion derived under i.i.d. Gaussian setting. 
A similar change-point statistic was considered by \cite{horvath2012},
\beqs
\cTh_{1, T} = \max_{b \in [1, T)} \frac{1}{\sqrt{n}}\frac{b(T-b)}{T^2}\sum_{j=1}^n\{(\cX^j_{1, b, T})^2 - 1\},
\eeqs
and the limit distribution of the test statistic was derived for independent panel data. 
\cite{enikeeva2014} proposed a change-point test for panel data with i.i.d. Gaussian noise
that combines the {\em linear} statistic, which is constructed similarly as $\cTh_{1, T}$, 
and the {\em scan} statistic, which is aimed at detecting a cross-sectionally sparse change.
More specifically, the test statistics are 
\beq
\cTlin_{1, T}  &=& \max_{b \in [1, T)} \frac{1}{H\sqrt{2n}} \sum_{j=1}^n\{(\cX^j_{1, b, T})^2 - 1\},  \mbox{ and } \nonumber
\\
\cTscan_{1, T}  &=& \max_{b \in [1, T)} \max_{1 \le m \le n} \frac{1}{T_m\sqrt{2m}} \sum_{j=1}^m\{(\cX^{(j)}_{1, b, T})^2 - 1\}, \nonumber
\eeq
where $|\cX^{(1)}_{1, b, T}| \ge \ldots \ge |\cX^{(n)}_{1, b, T}|$ denote the ordered CUSUM statistics at each $b \in [1, T)$. 
With $H$ and $T_m$ acting as critical values (chosen as approximate quantiles of $\chi^2$-distributions),
their proposed change-point test is $\cTeh = (\cTlin_{1, T} > 1) \vee (\cTscan_{1, T} > 1)$ (where $a \vee b = \max(a, b)$).

Allowing for both temporal and cross-sectional dependence, \cite{jirak2014} proposed a test statistic obtained from 
taking the pointwise maximum of the multiple CUSUM series:
\beq
\cTj_{1, T} = \max_{b \in [1, T)}\max_{1 \le j \le n}\sqrt{\frac{b(T-b)}{T}}|\cX^j_{1, b, T}|, \label{eq:jirak}
\eeq
which is compared against a threshold drawn from an extreme value distribution of Gumble type or bootstrap.

Note that with the exception of $\cTscan_{1, T}$, 
the CUSUM aggregation methods above are not adaptive to the underlying structure of CUSUM statistic values at each $b$,
in the sense that they take either pointwise maximum or sum of (squared) CUSUMs.
Empirical studies conducted in \cite{cho2015} showed that such approaches may lead to inferior performance
in detecting and locating change-points in high-dimensional settings. 
Instead, they proposed the Sparsified Binary Segmentation (SBS) where the change-point test
$\cTs_{1, T}(\pi_T)  > 0$ was based on the following ``sparsified'' or ``thresholded'' test statistic
\beqs
\cTs_{1, T}(\pi_T) = \max_{b \in [1, T)} \sum_{j=1}^n |\cX^j_{1, b, T}| \cdot \bbI(|\cX^j_{1, b, T}| > \pi_T)
\eeqs
($\bbI(\cE) = 1$ if and only if the event $\cE$ is true),
with an appropriately bounded threshold $\pi_T$ chosen to guarantee that
$|\cX^j_{s, b, e}| < \pi_T$ uniformly over $j \in \{1, \ldots, n\}$ and $\{(s, b, e); \, 1 \le s \le b < e \le T\}$ 
with probability converging to one under the null hypothesis of no change-point.
The intuition behind the construction of $\cTs_{1, T}(\pi_T)$ is that, 
irrelevant contribution from those components without any change 
is reflected as small values of $|\cX^j_{1, b, T}|$ and hence is disregarded through the thresholding step, 
while large CUSUMs formed in the vicinity of the true change-points
are summed up to strengthen the meaningful contribution from the corresponding components. 
Combining the BS procedure with the thresholded test statistic, 
the consistency of the SBS algorithm for multiple change-point detection was established.
However, it is empirically observed that the higher autocorrelations in $\vep_{j, t}$ is, 
the larger $\pi_T$ is required to grant the consistency in change-point detection,
which amounts to selecting $n$ thresholds for segmenting $n$-dimensional panel data. 


\subsection{Outline of the paper}
\label{sec:rest}

In this paper, we propose the {\em double} CUSUM statistic which
accomplishes the high-dimensional CUSUM series aggregation
through data-driven partitioning of the panel data at each point,
while avoiding the difficulties involved in selecting (possibly $n$) thresholds 
that apply directly to individual CUSUMs for establishing change-point detection consistency. 

The rest of the paper is organised as follows.
We describe the double CUSUM statistic in details in Section \ref{sec:method}.
In Section \ref{sec:theor}, we establish the consistency of the change-point test
based on the double CUSUM statistics, as well as investigating its efficiency
in comparison with the tests discussed in Section \ref{sec:lit}.
Also, the double CUSUM Binary Segmentation algorithm is formulated and
its consistency in multiple change-point detection is studied. 
Section \ref{sec:choice} discusses the choice of important quantities 
including the test criterion, for which a bootstrap procedure is introduced.
We illustrate its performance on simulated datasets in Section \ref{sec:sim} and 
on log returns of S\&P 100 component stock prices in Section \ref{sec:real}.
Section \ref{sec:conc} concludes the paper and the proofs of theoretical results are provided in Section \ref{sec:pf}.
Finally, some auxiliary results and simulation results are reported in the supplementary document \citep{cho2016}.

\section{Double CUSUM statistic}
\label{sec:method}

Recall the panel data model from Introduction
\begin{eqnarray}
\label{eq:panel}
x_{j, t} = f_{j, t} + \vep_{j, t}, \quad t=1, \ldots, T; \ j=1, \ldots, n.
\end{eqnarray}
The noise $\{\vep_{j, t}\}_{t=1}^T$ satisfies 
$\E(\vep_{j, t}) = 0$ for all $j$ and $t$,
and is allowed to be correlated both within-series and cross-sectionally as specified later in Section \ref{sec:theor}. 
The piecewise constant signals $\{f_{j, t}\}_{t=1}^T, \, j=1, \ldots, n$
share $N$ change-points $1<\eta_1 < \ldots < \eta_N<T$ (possibly with unknown $N$).
That is, at each change-point $\eta_r$, there exists an index set 
$\Pi_r = \{j:\, \delta_{j, r} = f_{j, \eta_r+1} - f_{j, \eta_r} \ne 0\} \subset \{1, \ldots, n\}$
with $m_r = |\Pi_r| = \sum_{j=1}^n \bbI(|\delta_{j, r}|>0) \ge 1$
(where $|\mathcal{S}|$ denotes the cardinality of a set $\mathcal{S}$).

Recall the definition of $\cX^j_{s, b, e}$ in (\ref{eq:cusum:two}), which denotes the CUSUM statistic computed on $x_{j, t}$ over a generic interval $t \in [s, e]$ with $1 \le s < e \le T$, for $b = s, \ldots, e-1$.
For change-point detection in panel data, we propose to employ the double CUSUM (DC) statistics 
\begin{eqnarray}
&& \cD_m^\varphi(\{|\cX^{(j)}_{s, b, e}|\}_{j=1}^n) \nonumber \\
&=&
\left\{\frac{m(2n-m)}{2n}\right\}^\varphi
\left(\frac{1}{m} \sum_{j=1}^m |\cX^{(j)}_{s, b, e}| -
\frac{1}{2n-m} \sum_{j=m+1}^n |\cX^{(j)}_{s, b, e}|\right)
\label{double:cusum:one}
\\
&=&
\left\{\frac{m(2n-m)}{2n}\right\}^\varphi \cdot
\frac{1}{m} \sum_{j=1}^m \left(|\cX^{(j)}_{s, b, e}| - \frac{1}{2n-m}\sum_{j=m+1}^n |\cX^{(j)}_{s, b, e}|\right)
\label{double:cusum:two}
\end{eqnarray}
for $b \in [s, e)$ and $m \in \{1, \ldots, n\}$,
where the DC operator $\cD^\varphi_m$ takes the {\em ordered} CUSUM values 
$|\cX^{(1)}_{s, b, e}| \ge |\cX^{(2)}_{s, b, e}| \ge \ldots \ge |\cX^{(n)}_{s, b, e}|$
at each $b$, as its input for some $\varphi\in[0, 1]$.
Then, the test statistic for detecting the presence of a change-point over a given interval $[s, e]$, is derived as
\beqs
\cT^\varphi_{s, e} = \max_{b\in[s, e)}\max_{1 \le m \le n} \cD_m^\varphi(\{|\cX^{(j)}_{s, b, e}|\}_{j=1}^n),
\eeqs
which is compared against a test criterion $\pi^\varphi_{n, T}$.
Once $\cT^\varphi_{s, e} > \pi^\varphi_{n, T}$,
the location of the change-point is identified as where the pointwise maximum (over $m$) of the DC statistics is maximised (over $b$), i.e.,
$$ \heta = \arg\max_{b\in[s, e)}\max_{1 \le m \le n} \cD_m^\varphi(\{|\cX^{(j)}_{s, b, e}|\}_{j=1}^n). $$

To understand the properties of the DC statistic, first consider the case where
the noise $\vep_{j, t}$ is not present in the panel data.
Then the series of DC statistics at each fixed $m$ is always maximised at one of the true change-points within the interval $[s, e)$
and consequently, the maximum over both $b$ and $m$ is guaranteed to be attained at a true change-point.
The formal statement and its proof can be found in Appendix B of the supplementary document.

The key feature of the DC statistic is the ordering of the input series to $\cD^\varphi_m$.
In panel data segmentation, in addition to detecting and locating the change-points in time,
it is also of interest to locate the change in coordinates as well,
which is relatively unexplored in the relevant literature with the exception of \cite{jirak2014}.
Not only such information is useful in the interpretation of detected change-points, 
but also can play an important role in aggregating the high-dimensional CUSUM series efficiently as detailed below.

At a given time point $b$, one way of partitioning the components into those with changes and those without,
is to arrange the modulus of CUSUMs in the decreasing order,
and then to label the components which correspond to the first $m_b$ ($\in \{1, \ldots, n\}$) largest values of $|\cX^{j}_{s, b, e}|$ 
as being likely to have a change-point around $b$.
Note that aggregating the (squared) CUSUMs via pointwise averaging or maximising
implicitly takes $m_b = n$ or $m_b = 1$, respectively.
In constructing $\cTs(\pi_T)$, the choice of $m_b$ is associated with the choice of $\pi_T$, 
i.e., $m_b(\pi_T) = \sum_{j=1}^n\bbI(|\cX^j_{s, b, e}| > \pi_T)$.
The DC statistic provides a data-driven alternative for selecting $m_b$, namely
\beq
\wh{m}^\varphi_b = \arg\max_{1 \le m \le n} \cD^\varphi_m(\{|\cX^{(j)}_{s, b, e}|\}_{j=1}^n). \label{eq:wh:m}
\eeq
While we do not claim that the thus-obtained partitioning 
consistently identify $\Pi_r$ at each detected change-point,
we see in simulation studies reported in Section \ref{sec:sim} that 
the DC statistic performs well in change-point detection, 
by identifying those components that {\em contribute} to change-point detection according to the above $\wh{m}^\varphi_b$.

The RHS of (\ref{double:cusum:two}) indicates that the term $(2n-m)^{-1} \sum_{j=m+1}^n |\cX^{(j)}_{s, b, e}|$
acts as a threshold on $|\cX^{(j)}_{s, b, e}|, \ j=1, \ldots, m$ at each $b$,
and therefore $\cD_m^\varphi(\{|\cX^{(j)}_{s, b, e}|\}_{j=1}^n)$ is essentially
a scaled average of $m$ largest CUSUMs after soft-thresholding,
from which we draw resemblance between $\cT^\varphi_{s, e}$ and $\cTs_{s, e}(\pi_T)$. 
However, since $\cT^\varphi_{s, e}$ involves maximisation over both $b$ and $m$,
we avoid explicitly selecting $\pi_T$ that applies directly to individual CUSUMs while still enjoying the ``sparsifying'' effect of the thresholding step.

Comparing (\ref{eq:cusum:two}) and (\ref{double:cusum:one}),
we can see the connection between the two operators $\cD_m^\varphi$ and $\cC_b$ especially when $\varphi= 1/2$,
as both return the scaled difference between partial averages over two disjoint intervals. 
However, $\cD_m^\varphi$ involves $\{m(2n-m)/(2n)\}^\varphi$ instead of $\{m(n-m)/n\}^\varphi$,
although the input series is of length $n$.
This difference comes from the observation that 
the latter scaling factor may not favour a change-point that is shared by more than $[n/2]$ rows,
and even act as a penalty when $m_r$ is close to $n$,
which is against our intuition on the detectability of a change-point.
By adopting $\{m(2n-m)/(2n)\}^\varphi$,
the DC statistic can be regarded as being computed on the panel data of dimension $2n$,
where there are additional $n$ ``null'' components which are known to have no change-point.
Since $\{m(2n-m)/(2n)\}^\varphi$ is non-decreasing in $m$,
when considering the asymptotic efficiency of the DC statistic-based change-point test,
the cross-sectional size of change accounts for both $|\delta_{j, r}|$, the magnitude of jumps at the change-point,
and $m_r$, its ``density'' (as opposed to the sparsity) across the panel (see Remark \ref{rem:aston:two} for further discussion).

For illustration, we computed the DC statistics from panel data generated with a single change-point of different configurations
with $\varphi$ chosen as detailed in Section \ref{sec:varphi}.
Fixing $n = 250$ and $T = 100$, $\vep_{j, t}$ was simulated as in (N1) of Section \ref{sec:sim:single:model} with $\varrho = 0.2$
(which controls the degree of cross-correlations in $\vep_{j, t}$).
The piecewise constant signals were generated with a change-point at $\eta_1 = T/2$,
where $m_1$ out of $n$ components contained a shift of magnitude randomly drawn from a uniform distribution $\cU(0.75\delta_1, 1.25\delta_1)$
with $(m_1, \delta_1) = ([\log\,n], 0.24)$ and $([0.5n], 0.05)$.
We chose $m_1$ and $\delta_1$ in order to set $\sum_{j \in \Pi_1} \delta_{j, 1}^2$ at approximately the same level.

As shown in Figure \ref{fig:comparison}, the location of the true change-point in time
was accurately identified as where the pointwise maximum of DC statistics was maximised in both settings,
i.e., $\heta_1 = \arg\max_{b\in[1, T)}\max_{1 \le m \le n}\cD^\varphi_m(\{|\cX^{(j)}_{1, b, T}|\}_{j=1}^n)$.
Comparing the two heat maps, where $\cD^\varphi_m(\{|\cX^{(j)}_{1, \heta_1, T}|\}_{j=1}^n)$ was maximised over $m$, 
namely $\wh{m}_1 = \arg\max_{1 \le m \le n} \cD^\varphi_m(\{|\cX^{(j)}_{1, \heta_1, T}|\}_{j=1}^n)$,
was closer to $m_1$ for the larger $\delta_1$. 
This implies that not all components with the changes 
contribute to the detection of a change-point in the presence of noise,
due to small magnitude of the changes,
and using only a subset of $\Pi_1$ may serve the purpose better for change-point detection.

\begin{figure}
\centering
\includegraphics[scale=.25]{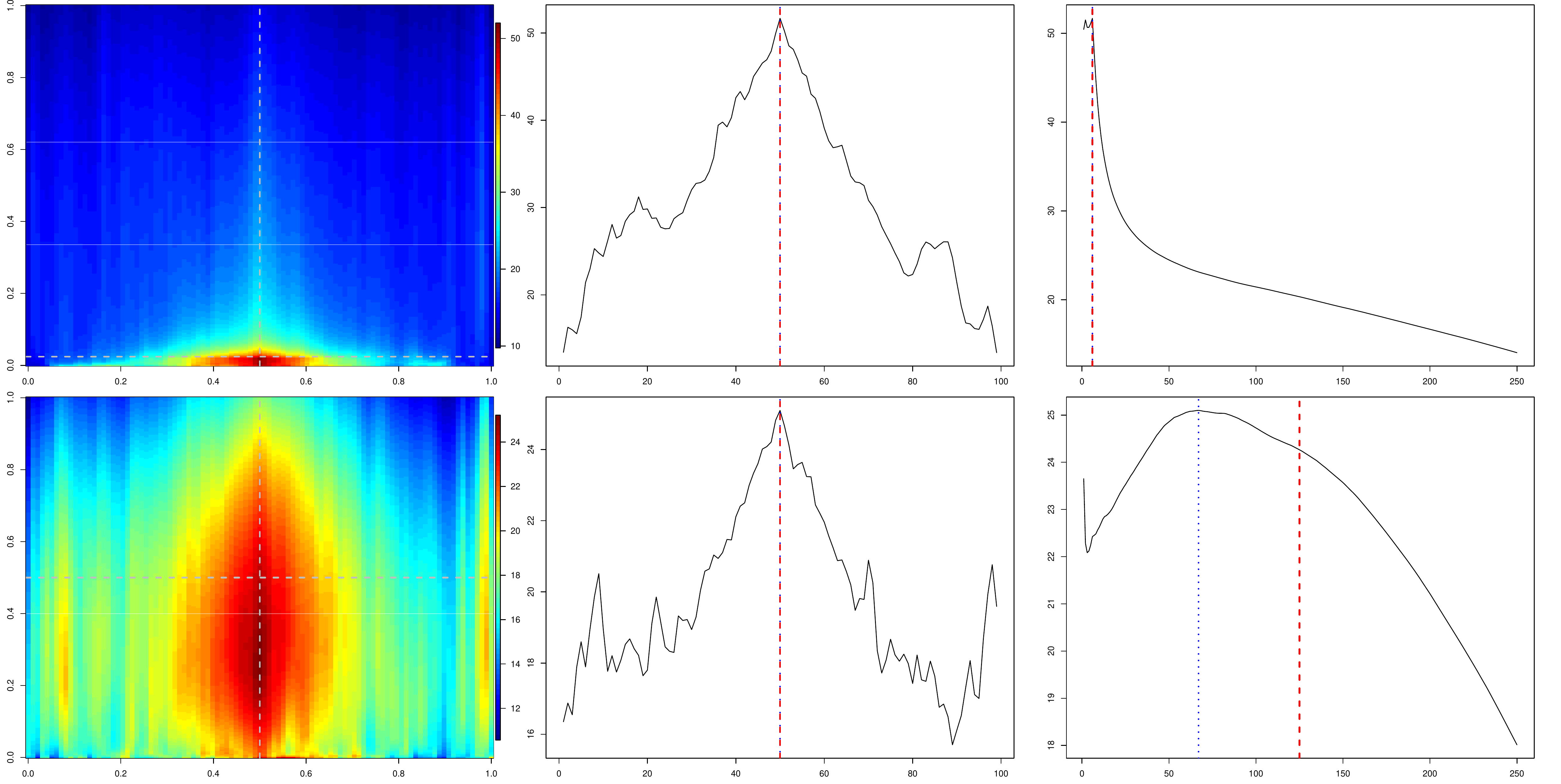}
\caption{$(m_1, \delta_1) = ([\log\,n], 0.24)$ (top) and $(m_1, \delta_1) = ([0.5n], 0.05)$ (bottom);
the heat map of $\cD^\varphi_m(\{|\cX^{(j)}_{1, b, T}|\}_{j=1}^n)$ over $b$ ($x$-axis) and $m$ ($y$-axis) (left),
the pointwise maximum of $\cD^\varphi_m(\{|\cX^{(j)}_{1, b, T}|\}_{j=1}^n)$ over $m$ for $b \in [1, T)$,
with broken lines indicating $\eta_1$ and the dotted ones $\heta_1$ (middle), and
$\cD^\varphi_m(\{|\cX^{(j)}_{1, \boldsymbol{\heta_1}, T}|\}_{j=1}^n), \, 1 \le m \le n$
with broken lines indicating $m_1$ and dotted ones $\wh{m}_1$.}
\label{fig:comparison}
\end{figure}

\vspace{5pt}
\begin{rem}
\label{rem:aston:one}
For high-dimensional change-point analysis, \cite{aston2014} defined the {\em high-dimensional efficiency},
a concept closely related to asymptotic relative efficiency.
Let $\bdel = (\delta_{1, 1}, \ldots, \delta_{n, 1})^\top$.
Then, the high-dimensional efficiency is determined by the rate at which 
the cross-sectional size of change is allowed to converge to zero ($\Vert \bdel \Vert_2 \to 0$) as $T$ and, with it, $n$ increase,
such that the power of the change-point test is strictly between the size and one.
They further investigated the problem of single change-point detection
using a class of change-point statistics
obtained from (i) first projecting the panel data with respect to a projection vector $\bp \in \R^n$,
and (ii) computing the series of CUSUM statistics from the univariate series $\{\langle \bx_t, \bp \rangle\}_{t=1}^T$,
with $\bx_t = (x_{1, t}, \ldots, x_{n, t})^\top$.
Assuming that $\bvep_t = (\vep_{1, t}, \ldots, \vep_{n, t})^\top$ is independent over time and denoting $\Sigma_{\vep} = \var(\bvep_t)$,
the oracle projection vector with the optimal high-dimensional efficiency was given by maximising
$(\bp^\top\Sigma_{\vep}\bp)^{-1}|\langle \bdel, \bp \rangle|^2$ with respect to $\bp$,
as $\mathbf{o} := \Sigma_{\vep}^{-1}\bdel$.

The DC statistic at fixed $b$ and $m$ coincides with the CUSUM of pointwise projection of the panel data.
More specifically, $\cD^\varphi_m(\{|\cX^{(j)}_{1, b, T}|\}_{j=1}^n)$ is associated with the $n$-vector $\bp^\varphi_{b, m}$
with its elements
\beqs
[\bp^\varphi_{b, m}]_j = \left\{\begin{array}{ll}
\sgn(\cX^j_{1, b, T}) \cdot \sigma_j^{-1}\left\{\frac{m(2n-m)}{2n}\right\}^\varphi\frac{1}{m} & \mbox{if } |\cX^j_{1, b, T}| \ge |\cX^{(m)}_{1, b, T}|, \\
-\sgn(\cX^j_{1, b, T}) \cdot \sigma_j^{-1}\left\{\frac{m(2n-m)}{2n}\right\}^\varphi\frac{1}{2n-m} & \mbox{if } |\cX^j_{1, b, T}| < |\cX^{(m)}_{1, b, T}|,
\end{array}\right.
\eeqs
such that $\cD^\varphi_m(\{|\cX^{(j)}_{1, b, T}|\}_{j=1}^n) = \cC_b(\{\langle \bx_t, \bp^\varphi_{b, m} \rangle\}_{t=1}^T)$.
In light of the previous discussion on pointwise ordering and partitioning of the panel data in computing the DC statistics, 
our approach may be viewed as an attempt at mimicking the performance of the oracle projection
without the prior knowledge of either $\bdel$ or $\Sigma_\vep$.
We further investigate the high-dimensional efficiency of the DC statistic in comparison with other competitors and the oracle projection
in Section \ref{sec:single}.
\end{rem}

\vspace{5pt}
\begin{rem}
\label{rem:ef}
The scan statistic \citep{enikeeva2014} shares with the DC statistic
the pointwise maximisation of {\em cumulative sums of the ordered CUSUMs} over $1 \le m \le n$.
One key difference is that the former has $(\cX^{(j)}_{1, b, T})^2$ in place of $|\cX^{(j)}_{1, b, T}|$,
with the interpretation of being the marginal log-likelihood ratio at given $b$ and $m$
under the assumption of i.i.d. Gaussian noise.
On the other hand, the latter can be seen as a projected change-point statistic as noted in Remark \ref{rem:aston:one}.
Another difference comes from the fact that $T_m$, acting as thresholds for $\cTscan_{s, e}$, are dependent on $m$,
while the test criterion compared against $\cT^\varphi_{s, e}$ does not depend on the choice of $m$.

\cite{enikeeva2014} proposed two choices for $T_m$, 
a theoretical one from a $\chi^2_m$-distribution,
and an empirical one $T_m = 2(2n)^{-1/2}\{m\log(ne/m)+\log(Tn/\alpha)\}$ 
at a given significance level $\alpha\in(0, 1)$.
In conducting the simulation studies reported in Section \ref{sec:sim},
both choices of $T_m$ were observed to be too sensitive to cross-sectional correlations in the panel data.
$T_m$ may be chosen numerically from e.g., a bootstrap scheme as indicated by the authors,
but this requires selecting $n$ such $T_m$ and moreover, 
it is unclear how the cross-correlations should be treated.
\end{rem}

\section{Consistency of the double CUSUM statistic}
\label{sec:theor}

\subsection{Single change-point detection}
\label{sec:single}

In this section, we show the consistency of the DC statistic 
in the single change-point scenario ($N = 1$). 
More specifically, we consider the null hypothesis
\beqs
\mathcal{H}_0: \, f_{j, 1} = \cdots = f_{j, T} \quad \mbox{for all} \quad j = 1, \ldots, n, 
\eeqs
which indicates structural stability in the mean over time.
As an alternative hypothesis, we specify the scenario 
where the piecewise constant signals $\{f_{j, t}\}_{t=1}^T, \, j=1, \ldots, n$
contains a single change-point at an unknown location $\eta_1 \in (1, T)$, such that
\beqs
\mathcal{H}_A: \, f_{j, 1} = \cdots = f_{j, \eta_1} \ne f_{j, \eta_1+1} = \cdots = f_{j, T} \, \mbox{for some} \,
j \in \Pi_1 \, \mbox{with} \, m_1 = |\Pi_1| \ge 1.
\eeqs

Throughout the paper, $a \sim b$ is used to denote that $a$ is of the order of $b$,
and $a \wedge b = \min(a, b)$.
Then, the consistency of the proposed test is established under the following conditions.

\begin{itemize}
\item[(A1)] For each $j$, $\vep_{j, t}$ denotes a stationary, zero-mean process with the mixing coefficients
\beqs
\alpha_j(k) = \sup_{\substack{G \in \sigma(\vep_{j, t+k}, \vep_{j,  t+k+1}, \ldots) \\ 
H \in \sigma(\vep_{j, t}, \vep_{j, t-1}, \ldots)}} |\pr(H \cap G) - \pr(H)\pr(G)|.
\eeqs
Then, there exist fixed $C_\vep, C_\alpha > 0$, $\mu\in (0, 1)$ and $\sigma^* \ge \sigma_* > 0$ with $\bar{\sigma} = \sigma^*/\sigma_* \in [1, \infty)$,
such that the followings hold.
\begin{itemize}
\item[(A1.i)] $\E(\vep_{j, t}^k) \le C_{\vep}^{k-2} k! \E(\vep_{j, t}^2)$ uniformly in $j$ and $k=3, 4, \ldots$.
\item[(A1.ii)] $\max_{1 \le j \le n}|\alpha_j(k)| \le C_\alpha \mu^k$ for all $k = 1, 2, \ldots$.
\item[(A1.iii)] $\sigma_j^2$, the long-run variance of $\vep_{j, t}$, satisfies $\sigma_*^2 \le \sigma_j^2 \le \sigma^{*2}$.
\end{itemize}
\item[(A2)] The dimensionality $n$ satisfies $n \sim T^\omega$ for some fixed $\omega\in[0, \infty)$.
\item[(A3)] There exists a fixed constant $\bar{f}>0$ such that 
$\max_{1 \le j \le n}\max_{1 \le t \le T} |f_{j, t}| \le \bar{f}$.
\item[(A4)] There exists a fixed constant $c>0$ such that $\eta_1 \wedge (T-\eta_1) > cT^\beta$
for some $\beta\in(0, 1]$.
\item[(A5)] Let $\wt{\delta}_1 = m_1^{-1}\sum_{j \in \Pi_1} |\delta_{j, 1}|$,
i.e., the average magnitude of non-zero changes at $t = \eta_1$. 
Then for any $\varphi\in[0, 1]$, we have
$(n^\varphi\log\,T)^{-1}m_1^\varphi\wt{\delta}_1 T^{\beta/2} \to \infty$ as $T \to \infty$.
\end{itemize}
In (A1), each $\vep_{j, t}$ is assumed to be $\alpha$-mixing (strong-mixing) at a geometric rate, with bounded moments and long-run variance.
The condition (A1.i) is met by many distributions such as exponential, gamma and inverse Gaussian, besides the Gaussian distribution.
Note that the cross-sectional correlations of the panel data are not explicitly controlled by any of the conditions imposed in (A1).
In (A2), the dimensionality can either be fixed or increase with $T$ at a polynomial rate.
(A4) imposes a condition on the unbalancedness of the change-point location,
permitting $T^{-1}\{\eta_1 \vee (T-\eta_1)\} \to 0$ as $T \to \infty$ when $\beta<1$.
(A5) places a lower bound on the rate of $m_1^\varphi\wt{\delta}_1$, which
dictates the minimum requirement on the cross-sectional size of the change
for the change-point to be detected as well as being located with accuracy.
(A3) rules out the trivial case where any $|\delta_{j, 1}| \to \infty$ with $T \to \infty$.

\vspace{5pt}
\begin{rem}[High-dimensional efficiency]
\label{rem:aston:two}
Recall the high-dimensional efficiency discussed in Remark \ref{rem:aston:one},
which was introduced in \cite{aston2014} as a tool that allows us
to quantify and compare the power of different change-point tests.
Table \ref{table:high} summarises the high-dimensional efficiency of $\cT^\varphi_{1, T}$ and other change-point tests discussed in Section \ref{sec:lit},
when the change-point is maximally distanced from the extreme ends of $[1, T]$ (i.e., $\beta = 1$).
The double vertical line divides the tests into 
(i) those proposed under the assumption of cross-sectional independence (left) and dependence (right),
(ii) those detecting the presence of a change-point only (left) and those identifying its location as well (right),
and (iii) those with the interpretation as projected change-point tests (right) and those without (left).

The oracle projection-based change-point test (see Remark \ref{rem:aston:one}) achieves high-dimensional efficiency of
$T^{1/2}\Vert \Sigma_{\bvep}^{-1/2}\bdel \Vert_2 \to \infty$ 
and thus attaining better efficiency than $\cTh_{1, T}$ and $\cTlin_{1, T}$ by $n^{1/4}$,
and than $\cTscan_{1, T}$ by $m_1^{1/2}$,
for diagonal $\Sigma_{\bvep}$ (independent panel).
When the change-point is sparse ($m_1 \sim 1$), the high-dimensional efficiency of $\cT^0_{1, T}$ 
is comparable to that of the oracle up to a logarithmic factor.
When the change-point is dense ($m_1 \sim n$),
the high-dimensional efficiency of $\cT^\varphi_{1, T}, \, \varphi>0$ is comparable to that of $\cT^0_{1, T}$.
Comparing $\cTj_{1, T}$, $\cTs_{1, T}(\pi_T)$ and $\cT^0_{1, T}$, 
the latter attains better high-dimensional efficiency when $m_1^{-1}\sum_{j\in\Pi_1}|\delta_{j, 1}| \gg \min_{j\in\Pi_1}|\delta_{j, 1}|$. 
However, the former two achieves partitioning consistency (i.e., consistent estimation $\Pi_1$),
which is not granted by the latter.
\begin{table}[htbp]
\caption{High-dimensional efficiency of change-point tests when $\beta = 1$.}
\label{table:high}
\centering
\begin{tabular}{c|l||c|l}
\hline
$\cTh_{1, T}$ & $\displaystyle{\frac{(\sum_{j=1}^n\delta_{j, 1}^2)^{1/2}T^{1/2}}{n^{1/4}} \to \infty}$ 
& $\cTj_{1, T}$ & $\displaystyle{\frac{\min_{j\in\Pi_1}|\delta_{j, 1}| T^{1/2}}{\sqrt{\log\,T}} \to \infty}$ \\
\hline
$\cTlin_{1, T}$ & $\displaystyle{\frac{(\sum_{j=1}^n\delta_{j, 1}^2)^{1/2}T^{1/2}}{n^{1/4}} \to \infty}$ 
& $\cTs_{1, T}(\pi_T)$ & $\displaystyle{\frac{\min_{j\in\Pi_1}|\delta_{j, 1}| T^{1/2}}{\log\,T} \to \infty}$ \\
\hline
$\cTscan_{1, T}$ & $\displaystyle{\frac{(\sum_{j=1}^n\delta_{j, 1}^2)^{1/2}T^{1/2}}{\sqrt{m_1\log(n/m_1)}} > \sqrt{6.6}}$ 
& $\cT^\varphi_{1, T}$ & $\displaystyle{\frac{m_1^\varphi\wt\delta_1 T^{1/2}}{n^\varphi\log\,T} \to \infty}$ \\
\hline
\end{tabular}
\end{table}
\end{rem}

\begin{rem}
\label{rem:scaling}
As briefly noted below (\ref{eq:cusum:two}), some studentisation is required for panel data analysis in practice. 
We estimate $\sigma_j$ using the flat-top kernel estimator with the automatically chosen bandwidth as discussed in \cite{politis2011}.
Let $\heta_{j, 1} = \arg\max_{b\in[1, T)}|\cC_b(\{x_{j, t}\}_{t=1}^T)|$,
\beq
\bar\vep_{j, t} &=& x_{j, t} - \frac{1}{\heta_{j, 1}}\sum_{t=1}^{\heta_{j, 1}}x_{j, t} \cdot \bbI(t \le \heta_{j, 1}) - 
\frac{1}{T-\heta_{j, 1}}\sum_{t=\heta_{j, 1}+1}^T x_{j, t} \cdot \bbI(t > \heta_{j, 1}), 
\label{eq:vep:est}
\\
w(t) &=& \left\{\begin{array}{ll}
1 & \mbox{for } |t| \le 1/2, \\
2(1-|t|) & \mbox{for } 1/2 < |t| < 1, \\
0 & \mbox{for } |t| \ge 1.
\end{array}\right. \nonumber
\eeq
and $c_j(k) = \frac{1}{T} \sum_{t=1}^{T-k} \bar\vep_{j, t} \bar\vep_{j, t+k}$. 
Note that $\bar\vep_{j, t}$ estimates the unobservable $\vep_{j, t}$
by identifying a change-point candidate $\wh{\eta}_{j, 1}$ associated with each $\{x_{j, t}\}_{t=1}^T$.
Let $\tau_j$ be the smallest positive integer such that
$|c_j(\tau_j+k)/c_j(0)| < 1.4\sqrt{T^{-1}\log_{10}\,T}$ for $k=1, 2, 3$ (where the constants are chosen as per \cite{huvskova2010}).
Then, the estimator is given by
\begin{eqnarray}
\wh\sigma_j^2 = \max\left\{c_j(0) + 2\sum_{k=1}^{2\tau_j}w\left(\frac{k}{2\tau_j}\right)c_j(k), \, \frac{c_j(0)}{2}\right\},
\label{eq:sig:est}
\end{eqnarray}
where the second term in the RHS of (\ref{eq:sig:est}) is present to 
prevent spuriously large CUSUMs resulting from too small estimates of $\sigma_j$
(flat-top kernel estimator is known to produce negative estimators).
\cite{huvskova2010} showed the consistency of the above estimator in the {\em single} change-point detection problem,
under similar conditions as those imposed on $\vep_{j, t}$ in (A1).
Noting that (i) the theoretical results derived in this paper hold without the consistency of the flat-top kernel estimator,
provided that $\wh\sigma_j^2$ is uniformly bounded away from zero, 
and also that (ii) we ultimately consider the problem of {\em multiple} change-point detection 
(with adjustment to the estimator of $\sigma_j^2$ in order to accommodate the presence of multiple change-points, see Remark \ref{rem:scaling:multi}),
we assume (A6) below on $\wh\sigma_j^2$ without extending the consistency result of \cite{huvskova2010}.
\end{rem}
\begin{itemize}
\item[(A6)] Define $\cE_\sigma = \{\max_{1 \le j \le n} |\wh\sigma_j^2 - \sigma_j^2| \le \sigma_*^2/2\}$ where $\sigma_*$ is the same as that in (A1).
Then we assume $\pr(\cE_\sigma) \to 1$ with $T \to \infty$. 
\end{itemize}

Note that under (A4), short intervals near the extreme points $\{1, T\}$ do not contain the change-point.
Hence we search for a change-point within $[1, T]\setminus\cI_{1, T}$ only, 
where $\cI_{1, T} = [1, d_T] \cup [T-d_T, T]$ with $d_T := [C\log^2\,T]$ for some $C > 0$.
That is,
$\cT^\varphi_{1, T} = \max_{b\in[1, T]\setminus\cI_{1, T}}\max_{1 \le m \le n} \cD^{\varphi}_m(\{|\cX^{(j)}_{1, b, T}|\}_{j=1}^n)$
and
$\heta_1 = \arg\max_{b\in[1, T]\setminus\cI_{1, T}}\max_{1 \le m \le n} \cD^{\varphi}_m(\{|\cX^{(j)}_{1, b, T}|\}_{j=1}^n)$.

Under these conditions, we present the following theorems on the consistency of the DC statistic-based test
equipped with a test criterion $\pi^\varphi_{n, T}$, which satisfies
$C'n^\varphi\log\,T < \pi^\varphi_{n, T} < C''m_1^\varphi\wt{\delta}_1 T^{\beta/2}$
for some $C', C'' > 0$.
\vspace{5pt}
\begin{thm}
\label{thm:zero}
Assume that (A1)--(A3) and (A6) hold and that there exists no change-point ($N=0$) in the panel data in (\ref{eq:panel}).
Then
$\pr\left\{\cT^\varphi_{1, T} > \pi^\varphi_{n, T}\right\} \to 0$
as $T \to \infty$.
\end{thm}
\vspace{5pt}
Theorem \ref{thm:zero} guarantees that when all signals remain constant,
the test does not detect any change-point with probability converging to one.
\vspace{5pt}
\begin{thm}
\label{thm:one}
Assume that (A1)--(A6) hold. Then there exists $c_0>0$ such that 
$$ \pr\left\{\cT^\varphi_{1, T} > \pi^\varphi_{n, T} \mbox{ and } |\heta_1-\eta_1| < c_0\ept\right\} \to 1 $$
as $T \to \infty$, where $\ept = n^{2\varphi}(m_1^\varphi\wt\delta_1)^{-2}\log^2\,T$.
\end{thm}
\vspace{5pt}
Theorem \ref{thm:one} states that in the presence of a single change-point,
the proposed test detects its presence as well as accurately identifying its location. 
From (A5), it is easily seen that $\ept/T^\beta \to 0$ as $T \to \infty$,
i.e., in the rescaled time, the estimated change-point location is consistent since 
$T^{-1}|\heta_1-\eta_1| \le T^{-\beta}|\heta_1-\eta_1| \to 0$.
We may define the optimality in change-point detection as
when the true change-point and its estimate are
within the distance of $O_p(1)$, see \cite{korostelev1987}.
Then, near-optimality in change-point estimation is achieved up to a logarithmic factor ($\log^2\,T$)
with the choice (i) $\varphi = 0$ when $\wt\delta_1 \sim 1$, or
(ii) $\varphi>0$ when $m_1 \sim n$ and $\wt\delta_1 \sim 1$. 

\subsection{Binary segmentation for multiple change-point detection}
\label{sec:bs}

In this section, we show the consistency of the DC statistic for multiple change-point detection
when applied jointly with a BS algorithm.
We first formulate the double CUSUM Binary Segmentation (DCBS) algorithm 
for panel data segmentation, which is equipped with a test criterion $\pi^\varphi_{n, T}$.
As in Section \ref{sec:single}, let $\cI_{s, e} = [s, s+d_T] \cup [e-d_T, e]$
denote a fraction of the interval $[s, e]$ on which we do not search for change-points,
in order to account for possible bias in the previously detected change-points.
This does not affect the asymptotic consistency of the estimated change-point locations 
under the assumption (B1) below, on the dispersion of change-points.
The index $p$ is used to denote the level (indicating the progression of the segmentation procedure)
and $q$ is used to denote the location of the node at each level.

\begin{description}
\item[The double CUSUM Binary Segmentation algorithm]
\item[Step 0]
Set $(p, q)=(1, 1)$, $s_{p, q}=1$ and $e_{p, q}=T$.

\item[Step 1] At the current level $p$, repeat the following for all $q$.
\begin{description}
\item[Step 1.1] Letting $s=s_{p, q}$ and $e=e_{p, q}$,
obtain the series of CUSUMs $\{\cX^{j}_{s, b, e}\}$ for $b\in[s, e)$ and $j=1, \ldots, n$,
on which $\cD^{\varphi}_m(\{|\cX^{(j)}_{s, b, e}|\}_{j=1}^n)$ is computed over all $b$ and $m$.
\item[Step 1.2] Obtain the test statistic $\cT^\varphi_{s, e} = \max_{b\in[s, e]\setminus\cI_{s, e}}\max_{1 \le m \le n} \cD^{\varphi}_m(\{|\cX^{(j)}_{s, b, e}|\}_{j=1}^n)$.
\item[Step 1.3] If $\cT^\varphi_{s, e} \le \pi^\varphi_{n, T}$,
quit searching for change-points on the interval $[s, e]$.
On the other hand, if $\cT^\varphi_{s, e} > \pi^\varphi_{n, T}$,
locate 
$$\heta = \arg\max_{b\in[s, e]\setminus\cI_{s, e}}\max_{1 \le m \le n} \cD^{\varphi}_m(\{|\cX^{(j)}_{s, b, e}|\}_{j=1}^n)$$
and proceed to Step 1.4.
\item[Step 1.4] Add $\heta$ to the set of estimated change-points
and divide the interval $[s_{p, q}, e_{p, q}]$ into two sub-intervals
$[s_{p+1, 2q-1}, e_{p+1, 2q-1}]$ and $[s_{p+1, 2q}, e_{p+1, 2q}]$,
where $s_{p+1, 2q-1} = s_{p, q}$, $e_{p+1, 2q-1} = \heta$, 
$s_{p+1, 2q} = \heta+1$ and $e_{p+1, 2q} = e_{p, q}$.
\end{description}

\item[Step 2]
Once $[s_{p, q}, e_{p, q}]$ for all $q$ are examined at level $p$,
set $p \leftarrow p+1$ and go to Step 1.
\end{description}

Step 1.3 furnishes a stopping rule to the DCBS algorithm:
the search for further change-points is terminated 
once $\cT^\varphi_{s, e} \le \pi^\varphi_{n, T}$ on every $[s, e]$ 
defined by two adjacent estimated change-points.

\vspace{5pt}
\begin{rem}
\label{rem:scaling:multi}
An adjustment is required to extend the scaling estimation procedure in Remark \ref{rem:scaling},
originally designed for single change-point detection,
to be applicable to the multiple change-point scenario.
More specifically, $\bar\vep_{j, t}$ in (\ref{eq:vep:est}) can no longer be regarded as well-estimating $\vep_{j, t}$
when $[1, \heta_{j, 1}]$ or $(\heta_{j, 1}, T]$ may contain further change-points. 
Seeing that $\heta_{j, 1}$ is the top node of a binary tree that results from 
applying a BS algorithm to $\{x_{j, t}\}_{t=1}^T$,
we remedy this by growing a binary tree of {\em depth} $L_T$ ($= O(\log\,T)$) 
from each $\{x_{j, t}\}_{t=1}^T$, $j=1, \ldots, n$. 
Such a tree represents a maximal segmentation of $x_{j, t}$ and therefore, regarding each segment as being stationary,
$\bar\vep_{j, t}$ is derived by subtracting the sample means within those intervals.
Then, the scaling estimation procedure is applied to the thus-obtained $\bar\vep_{j, t}$. 

Note that this approach requires an arbitrary choice of $L_T$
when there is no prior information on the upper bound on the total number of change-points, $N$.
For multiple change-point detection,
such a requirement (or its equivalent) is commonly found; see e.g., \cite{soh2014} and \cite{kirch2015}
where the minimum distance between two consecutive change-points plays an essential role
in guaranteeing the consistency of the proposed procedures.
\end{rem}
\vspace{5pt}

The consistency of the DCBS algorithm is established under (A1)--(A3), (A6)
and the following conditions extending (A4)--(A5) in order to allow for the presence of multiple change-points. 
\begin{itemize}
\item[(B1)] There exists a fixed constant $c>0$ such that
$\min_{0 \le r \le N} (\eta_{r+1} - \eta_r) \ge cT^\beta$ for some $\beta\in(6/7, 1]$,
using the notational convention that $\eta_0 = 0$ and $\eta_{N+1}=T$.
\item[(B2)] At each $\eta_r$,
let $\wt{\delta}_r = m_r^{-1}\sum_{j \in \Pi_r} |\delta_{j, r}|$.
Then for any $\varphi\in[0, 1]$, we have $\uDel_\varphi = \min_{1 \le r \le N} m_r^\varphi\wt{\delta}_r$ satisfy 
$(n^\varphi\log\,T)^{-1} \uDel_\varphi T^{7\beta/4-3/2} \to \infty$ as $T \to \infty$.
\end{itemize}

Since the unbalancedness of a change-point location is closely related to
the distance between two adjacent change-points, (B1) is formulated in terms of the latter.
Note that we do not impose any upper bound on the number of total change-points $N$
except for the implication that can be derived from (B1), namely that $N$ may grow with $T$ provided 
that any two adjacent change-points is sufficiently distanced.
Comparing (B2) to (A5), the high-dimensional efficiency is worsen by $T^{3/2-5\beta/4}$ when detecting multiple change-points,
which is also the case when applying the BS algorithm to univariate data.
This rate can be improved when the DC statistic is combined with e.g., the Wild Binary Segmentation (WBS)
\citep{piotr2014} which applies the CUSUM principle to randomly drawn intervals;
we leave the exploration in this direction for the future research. 

\vspace{5pt}
\begin{thm}
\label{thm:two}
Let $\heta_r, \ r=1, \ldots, \wh{N}$ ($1 < \heta_1 < \ldots < \heta_{\wh{N}} < T$) 
denote the change-points detected by the DCBS algorithm with a test criterion $\pi^\varphi_{n, T}$ satisfying
$C'n^{2\varphi}\uDel_\varphi^{-1}T^{5(1-\beta)/2}\log\,T < \pi^\varphi_{n, T} < C''\uDel_\varphi T^{\beta-1/2}$.
Assuming (A1)--(A3), (A6) and (B1)--(B2), there exists $c_0>0$ such that 
$$\pr\left\{\wh{N}=N; \,|\heta_r-\eta_r| < c_0\ept \mbox{ for } r=1, \ldots, N\right\} \to 1$$
as $T \to \infty$, where
$\ept = n^{2\varphi}\uDel_\varphi^{-2} T^{5(1-\beta)}\log^2\,T$.
\end{thm}
\vspace{5pt}
Unlike in the single change-point scenario, $\ept$ now depends on the dispersion of the change-points through $\beta$.
From (B1)--(B2), it is easily seen that $\ept/T^\beta \to 0$ as $T \to \infty$,
and hence the multiple change-points are consistently located for all $r=1, \ldots, N$.

\vspace{5pt}
\begin{rem}[Post-processing of estimated change-points]
\label{rem:post:proc}
In this paper, we focus on establishing the good performance of the DC statistic in high-dimensional CUSUM series aggregation
when combined with a simple segmentation method,
rather than considering more sophisticated algorithms such as the WBS \citep{piotr2014} or MOSUM \citep{kirch2015} procedures,
where the DC principle is straightforwardly applicable.
The BS algorithm is known to perform sub-optimally in certain unfavourable settings since,
at each iteration, it fits a step function with a single break to the data over a segment that
contains possibly multiple change-points.
Hence, we equip the DCBS algorithm with an extra step aiming at removing spuriously detected change-points
which is in line with the post-processing steps proposed in \cite{cho2012} and \cite{cho2015}.
It checks whether $\cT^{\varphi}_{\heta_r-d_r, \heta_r+d_r} > \pi^\varphi_{n, T}$
for $r=1, \ldots, \wh{N}$ with $d_r = \min(\heta_r-\heta_{r-1}, \heta_{r+1}-\heta_r)/2$.
In other words, we compute the DC statistics within each segment containing a single estimated change-point,
and retain only those $\heta_r$ that survive the thresholding.
More details on the post-processing step can be found in Section 3.2.1 of \cite{cho2012}.
\end{rem}

\section{Choice of $\varphi$ and test criterion}
\label{sec:choice}

\subsection{Choice of $\varphi$}
\label{sec:varphi}

Remark \ref{rem:aston:two} indicates that 
$\varphi  = 0$ is preferable in terms of the high-dimensional efficiency for 
detecting a change-point that is known to be sparse across the panel, 
while $\cT^\varphi_{1, T}$ with $\varphi > 0$ achieves the same high-dimensional efficiency as 
$\cT^0_{1, T}$ when the change is dense across the panel.
In practice, such information is often unavailable a priori and therefore it is of interest to
find a way of combining the information embedded in an {\em array} of DC statistics indexed by $\varphi \in  [0, 1]$,
which works universally well over different ranges of change-point configurations
determined by $\eta_r$, $m_r$ and $\wt{\delta}_r$.

Recalling the data-driven partitioning of the panel achieved by the DC statistic,
the pointwise partitioning of the ordered CUSUM values 
may be regarded as analogous to fitting a step function with a single break to $|\cX^{(j)}_{s, b, e}|, \ j=1, \ldots, n$ at each $b$. 
Suppose that $\{|\cX^{(j)}_{s, b, e}|\}_{j=1}^n$ arises from an additive model
with a piecewise constant signal that contains a single break at $j=m_r$ 
for some $r \in \{1, \ldots, N\}$ and otherwise constant.
This setting is unlikely to be satisfied by the ordered CUSUM values in general, 
but provides a framework for linking the optimality in panel partitioning to that in locating a single change-point.
Then, \cite{brodsky1993} showed that the choice of $\varphi = 1/2$ leads to 
the best rate of convergence in the bias $|\wh{m}^\varphi_b - m_r|$ (recall (\ref{eq:wh:m}))
among $\varphi \in [0, 1]$, i.e., asymptotically, the optimal partitioning is attained by the choice of $\varphi = 1/2$. 

Taking into account these observations, we propose a new DC statistic
\begin{eqnarray*}
\cDzh_m(\{|\cX^{(j)}_{s, b, e}|\}_{j=1}^n) = 
\gamma_n\cD^0_m(\{|\cX^{(j)}_{s, b, e}|\}_{j=1}^n) + \cD^{1/2}_m(\{|\cX^{(j)}_{s, b, e}|\}_{j=1}^n),
\end{eqnarray*}
where $\gamma_n$ acts as a scaling factor that enables
treating $\cD^0_m$ and $\cD^{1/2}_m$ on an equal footing.
Heuristically, the discrepancy between $\cD^0_m$ and $\cD^{1/2}_m$ increases at the rate $m^{1/2}$,
but is not so pronounced for small values of $m$.
Therefore, $\cDzh_m$ can be viewed as an attempt at combining the different ranges of change-point configurations 
over which either $\cD^0_m$ or $\cD^{1/2}_m$ achieves consistency,
and the following conditions modifying (A5) and (B2) reflect this point. 

\begin{itemize}
\item[(A5')] $(n^{1/2}\log\,T)^{-1}(\gamma_n \vee m_1^{1/2})\wt{\delta}_1 T^{\beta/2} \to \infty$ as $T \to \infty$.
\item[(B2')] We have $\uDel = \min_{1 \le r \le N} (\gamma_n \vee m_r^{1/2})\wt{\delta}_r$ satisfy 
$(n^{1/2}\log\,T)^{-1} \uDel T^{7\beta/4-3/2} \to \infty$ as $T \to \infty$.
\end{itemize}

Then, the consistency of DC statistics carry over to that of the newly defined DC statistic $\cDzh_m$ as below.
\begin{thm}
\label{thm:four} 
Assume $\omega > 0$ such that $n \to \infty$ as $T \to \infty$,
and $n^{-1/2}\gamma_n \to 0$ as $n \to \infty$. 
Let $\cDzh_m$ replace $\cD^\varphi_m$ in deriving the test statistics where applicable.
\begin{itemize}
\item[(a)] Assume that (A1)--(A4), (A5') and (A6) hold.
Equipped with the test criterion $\pizh_{n, T}$ satisfying 
$C'n^{1/2}\log\,T <$ $\pizh_{n, T}<$ $C''(\gamma_n \vee m_1^{1/2})\wt{\delta}_1 T^{\beta/2}$ for some $C', C'' > 0$,
Theorems \ref{thm:zero}--\ref{thm:one} hold with 
$\ept = n\{(\gamma_n \vee m_1^{1/2})\wt\delta_1\}^{-2}\log^2\,T$ in the latter.

\item[(b)] Assume that (A1)--(A3), (A6), (B1) and (B2') hold.
Equipped with the test criterion satisfying
$C'n\uDel^{-1}T^{5(1-\beta)/2}\log\,T <$ $\pizh_{n, T} <$ $C''\uDel T^{\beta-1/2}$ for some $C', C'' > 0$,
Theorem \ref{thm:two} holds with 
$\ept = n\uDel^{-2} T^{5(1-\beta)}\log^2\,T$.
\end{itemize}
\end{thm}

For the purpose of testing only, it is reasonable
to combine the two tests based on $\cD^0_m$ and $\cD^{1/2}_m$ as in \cite{enikeeva2014},
namely, $(\cT^0_{1, T} > \pi^0_{n, T}) \vee (\cT^{1/2}_{1, T} > \pi^{1/2}_{n, T})$.
Adopting such an approach, however, it is not clear how to identify the location of a change-point once its presence is detected,
since we cannot exclude the possibility that it is detected by both $\cT^0_{1, T}$ and $\cT^{1/2}_{1, T}$,
each of which estimates its location with bias.
This becomes further complicated in the presence of multiple change-points of different change-point configurations.
On the other hand, while $\gamma_n$ needs to be selected additionally,
we can readily identify the change-point locations using $\cDzh_m$.

Theorem \ref{thm:four} assumes that $n^{-1/2}\gamma_n \to 0$ only.
However, as the formulation of (A5') and (B2') suggests, 
guaranteed improvement of the high-dimensional efficiency of $\cDzh_m$ over that of both $\cD^0_m$ and $\cD^{1/2}_m$,
still requires the knowledge of unknown sparsity or density of the change-point for the choice of $\gamma_n$.
In Section \ref{sec:sim}, $\gamma_n = \log\,n$ is considered
in investigating the finite sample performance of $\cDzh_m$ along with that of $\cD^0_m$ and $\cD^{1/2}_m$. 

\subsection{Bootstrap for test criterion selection}
\label{sec:threshold}

Theorems \ref{thm:zero}--\ref{thm:four} provide ranges for the rate of $\pi^\varphi_{n, T}$ 
which grant size and power consistency for both single and multiple change-point detection problems.
However, the theoretical rates involve typically unattainable knowledge
on the quantities such as $\beta$ and $\uDel_\varphi$ and, even with such knowledge available,
finite sample performance may be heavily influenced by the choice of the multiplicative constant applied to the given rate.
Therefore, we propose a resampling procedure that enables us to 
approximate the quantiles of $\cT^\varphi_{s, e}$ under the null hypothesis of no change-points.

Originally proposed as a data-based simulation method for statistical inference with i.i.d. random samples,
the bootstrap principle has been extended to a variety of non-i.i.d. situations; see \cite{hardle2003} for a comprehensive survey. 
Although some heuristic attempts have been made \citep{fiecas2014}, 
applying bootstrap methods developed for time series of fixed dimensions to high-dimensional settings is challenging.
Some recent efforts in this direction include \cite{jentsch2015}, 
who proposed the Linear Process Bootstrap for multivariate time series
and established its asymptotic validity for the sample mean when $n$ is allowed to increase with $T$.
The procedure involves the estimation of $(nT)\times(nT)$-dimensional covariance matrix 
of the vectorised version of $[\bvep_t, \, t=1, \ldots, T]$ and the computation of its Cholesky decomposition,
where the latter task alone is of computational complexity $O(n^3T^3)$,
which makes it difficult to apply the method even to panel data of moderately large dimensions.

Instead, we propose a bootstrap procedure motivated by the representation theory developed for the Generalised Dynamic Factor Model (GDFM).
Factor analysis is a popular dimension reduction technique used in many disciplines, such as 
econometrics, statistics, signal processing, psychometrics, and chemometrics. 
The key idea is that pervasive cross-correlations in $\bvep_t$ are modelled by the {\em common} component $\bchi_t$,
and $\bxi_t = \bx_t - \bchi_t$ with moderate degree of cross-correlations denotes the {\em idiosyncratic} component. 
The GDFM introduced in \cite{forni2000} steps further and admits the following representation theorem \citep{forni2001}:
any $\bvep_t$ with a finite number ($q < \infty$) of diverging dynamic eigenvalues 
is decomposed into $\bchi_t$ driven by a $q$-tuple of white noises $\bu_t$ ({\em common shocks}) as $\bchi_t = \bb(L)\bu_t$
($L$ denotes the lag operator and $\bb(L)$ is an $n \times q$-matrix of one-sided and square-summable filters $b_{ik}(L)$), 
and $\bxi_t$ with bounded dynamic eigenvalues.
Referred to as the GDFM-boot, the proposed bootstrap method utilises this representation property of the GDFM
in order to effectively handle the cross-correlations as well as within-series correlations present in $\bvep_t$ of large dimensions.

\begin{description}
\item[The GDFM-boot algorithm]
\item[Step 1] Obtain 
$\bE = [\wh\bvep_t, \, t=1, \ldots, T]$ where $\wh\vep_{j, t} = \wh{\sigma}_j^{-1}\bar\vep_{j, t}$ (refer to (\ref{eq:vep:est}) and Remark \ref{rem:scaling:multi} for the definition of $\bar\vep_{j, t}$).

\item[Step 2] Let $\mathbf{e}_l$ denote the eigenvector of the covariance matrix $T^{-1}\bE\bE^\top$, which corresponds to its $l$-th largest eigenvalue. 
Then, estimate $q$, the number of common shocks driving the cross-correlations of $\bE$ 
using the information criterion proposed by \cite{bai2002}:
$IC(k) = \log(T^{-1}\sum_{t=1}^T \Vert\wh{\bvep}_t - \sum_{l=1}^k\mathbf{e}_l\mathbf{e}_l^\top\wh{\bvep}_t\Vert_2^2) + 
kC_{n, T}^{-1}\log(C_{n, T})$,
as $q = \arg\min_{0 \le k \le Q} IC(k)$ with $Q = \lfloor C_{n, T}/\log(C_{n, T}) \rfloor$ and $C_{n, T} = n \wedge T$.

\item[Step 3] Estimate the $q$-dimensional common shocks ($\wh{\bu}_t$) and the associated filter ($\wh\bb(L)$),
such that $\wh\bvep_t$ is decomposed into 
$\wh\bchi_t = \wh\bb(L)\wh\bu_t$ and $\wh\bxi_t = \wh\bvep_t - \wh\bchi_t$.

\item[Step 4] For $l=1, \ldots, B$, repeat the following steps. 
\begin{description}
\item[Step 4.1] For each $k=1, \ldots, q$, draw $\{u^l_{k, t}\}_{t=1}^T$
independently from the empirical distribution of $\{\wh{u}_{k, t}\}_{t=1}^T$,
from which $\bchi^l_t = \wh\bb(L)\bu^l_t$ is obtained.
\item[Step 4.2] Generate a bootstrap sequence of the Fourier coefficients of $\wh\bxi_t$
using the Time Frequency Toggle (TFT)-Bootstrap \citep{kirch2011}, to which the inverse fast Fourier transform algorithm 
is applied to produce $\bxi_t^l$.
\item[Step 4.3] On the bootstrap series $\bvep^l_t = \bchi^l_t + \bxi^l_t$, compute
$\cE^{j, l}_{1, b, T} = \cC_b(\{\vep^l_{j, t}\}_{t=1}^T)$,
from which
$\cT^{\varphi, l}_{1, T} = \max_{b\in[1, T), \, 1 \le m \le n}\cD^\varphi_m(\{|\cE^{(j), l}_{1, b, T}|\}_{j=1}^n)$
is obtained where $|\cE^{(1), l}_{1, b, T}| \ge \ldots \ge |\cE^{(n), l}_{1, b, T}|$. 
\end{description}

\item[Step 5] Select the $(1-\alpha)$-quantile of $\cT^{\varphi, l}_{1, T}, \, l=1, \ldots, B$ as $\pi^\varphi_{n, T}$.
\end{description}

In Step 1, we adopt the same approach taken in estimating $\sigma_j$,
namely first to estimate the locations of the change-points coordinate-wise and then 
to take away the estimated piecewise constant signal from each $x_{j, t}$,
in order to estimate the unobservable $\vep_{j, t}$. 
Note that the practice of input data studentisation is often adopted prior to factor modelling.
Step 2 is justified by the observations that 
(i) the number of factors in the static representation of factor models
is typically greater than the number of common shocks in the GDFM, and
(ii) over-estimated $q$ still returns mean-square consistent common components \citep[Corollary 2]{forni2000}.
Section 4 of \cite{forni2000} provides an estimator of $\bb(L)$ (and hence $\bu_t$) 
based on the windowed estimator of the spectral density matrix of $\bE$, which is adopted for Step 3. 

Step 4.1 produces a bootstrap sample of $\wh\bchi_t$ by treating the white noise estimates $\wh{u}_{k, t}$ as being i.i.d. over time.
In Step 4.2, the Local Bootstrap, originally proposed in \cite{kirch2011} as a part of TFT-Bootstrap for resampling univariate time series,
is applied to the $n$-dimensional $\wh\bxi_t$ of bounded cross-sectional correlations.
It does not require an initial estimate of the spectral density matrix of $\wh\bxi_t$.
Instead, the Fourier coefficients of $\wh\bxi_t$ are resampled according to a kernel-based probability vector,
under the observation that in a neighbourhood of each frequency, 
the distribution of different coefficients is almost identical (if the spectral density is smooth).
A detailed description of the Local Bootstrap is provided in the Appendix A of the supplement document.
In simulation studies, we used the Daniell kernel with the window size $[0.05T]$. 
Since the common and idiosyncratic components are handled independently, 
generating a smaller size bootstrap sample (e.g., $\lceil \sqrt{B}\rceil$) for each component may be sufficient
to generate the bootstrap sample of size $B$ for $\bvep_t$.

Once a bootstrap sample is generated at level $p=1$ of the DCBS algorithm,
the same sample can be used repeatedly for critical value selection at levels $p \ge 2$.
That is, in Step 4.3 above, the test statistics are computed as
$\cT^{\varphi, l}_{s', e'} = \max_{b\in[s', e'), \, 1 \le m \le n}\cD^\varphi_m(\{|\cE^{(j), l}_{s', b, e'}|\}_{j=1}^n)$
over a moving window of size $e-s+1$ 
for $1 \le s' \le T-e+s$, $e' = s'-s+e$ and $l=1, \ldots, B$,
from which the test criterion is drawn.

Establishing the validity of the GDFM-boot algorithm for change-point analysis 
is beyond the scope of the current paper and hence is left for the future investigation.
However, simulation studies in Section \ref{sec:sim} confirm its good performance 
in various settings together with the proposed change-point statistic.
Computationally, the GDFM-boot algorithm benefits from the dimension reduction via factor modelling.
When applied to generate a bootstrap sample of size $B = 100$ with $n = 25$ and $T = 100$
(executed on a $3.10$GHz quad-core with $8$GB of RAM running Windows $7$),
the \texttt{R} code implementing the algorithm took $0.38$ seconds,
compared to $62$ seconds taken by the Linear Process Bootstrap.
Further, applying the latter was not computationally feasible for any data of the size and dimensionality considered in our simulation study.

\section{Simulation study}
\label{sec:sim}

\subsection{Single change-point detection}
\label{sec:sim:single}

In this section, we evaluate the empirical performance 
of the DC statistic on simulated panel data with (at most) a single change-point.
For comparison, change-point tests with two different choices of $\varphi=0, 1/2$ (referred to as $\cT^0$ and $\cT^{1/2}$)
as well as the combined DC statistic ($\cTzh$) from Section \ref{sec:varphi} are considered,
with $\pi^\varphi_{n, T}$ computed from the GDFM-boot algorithm.
Also included are 
\begin{itemize}
\item $\cTeh_{1, T} = (\cTlin_{1, T} > 1) \vee (\cTscan_{1, T} > 1)$ equipped with the thresholds 
$H$ chosen as $(1-\alpha/2)$-quantile of the $\chi^2_n$-distribution and
$T_m = \frac{2}{\sqrt{2m}}\{m\log(ne/m)+\log(nT/\alpha)\}$,
as per the recommendation made in Section 5 of \cite{enikeeva2014}
(referred to as $\cTeh$);
\item $\cTj_{1, T}$ with the test criterion chosen according to 
the bootstrap algorithm (Algorithm 4.6) proposed in \cite{jirak2014} with the block size $K = 4$
(referred to as $\cTj$).
\end{itemize}
Since the scaling of component series is not discussed in \cite{enikeeva2014}, 
we choose to apply $\wh{\sigma}_j$ estimated as in Remark \ref{rem:scaling} in deriving $\cTeh_{1, T}$.
\cite{jirak2014} showed the consistency of the proposed test when applied with a set of long-run variance estimators for $\sigma_j$. 
As discussed in Section \ref{sec:sim:single:res}, $\cTj$ is highly sensitive to the scaling terms estimated spuriously small,
and therefore we use the most conservative estimator among those included in Proposition 3.5 of \cite{jirak2014}.
$\cTh_{1, T}$ is similarly constructed as $\cTlin_{1, T}$ with the same high-dimensional efficiency,
and $\cTeh$ is expected to perform better than either of the two component tests. Hence $\cTh_{1, T}$ is omitted from our study.

As described in Section \ref{sec:sim:single:model}, each coordinate of $\bvep_t$ is generated from the same model,
which enables us to select a single threshold $\pi_T$ applicable to all $n$ component series in deriving $\cTs_{1, T}(\pi_T)$.
Hence we include $\cTs_{1, T}(\pi_T)$ equipped with the ``oracle'' threshold in the comparative study (referred to as $\cTs$),
where $\pi_T$ is chosen from the GDFM-boot algorithm with (unobservable) $\bvep_t$ replacing the estimated $\bar{\bvep}_t$.
Then, $\pi_T$ is the $(1-\alpha)$-quantile of $\max_{b\in[1, T), \, 1 \le j \le n}|\cE^{j, l}_{1, b, T}|, \, l=1, \ldots, B$.

When a bootstrap procedure is involved for test criterion selection, the bootstrap sample size is fixed at $B = 100$.
We report the Type I error and the size-corrected power at the significance level $\alpha = 0.05$, 
as well as the location accuracy ($|\heta_1 - \eta_1| < \log\,T$, in $\%$) for $\cT^0$, $\cT^{1/2}$, $\cTzh$, $\cTj$ and $\cTs$
over $100$ simulated sample paths for each setting; for $\cTeh$, it is not suggested how the change-point location is to be estimated. 
For all tests, $d_T = 5$ is used to trim off the extreme ends of the interval $[1, T]$.
Also, we investigate the partitioning accuracy of $\cT^0$, $\cT^{1/2}$, $\cTzh$, $\cTj$ and $\cTs$ by reporting the Rand index,
the sum of true positives and true negatives divided by the total ($n$),
where Rand index close to $1$ indicates more accurate partitioning (binary classification).

\subsubsection{Data generating models}
\label{sec:sim:single:model}

We consider piecewise constant signals $\{f_{j, t}\}_{t=1}^{T}$
with varying sparsity/density ($m_1$), size of jumps ($|\delta_{j, 1}| \sim \cU(0.75\delta_1, 1.25\delta_1)$ 
for $j\in\Pi_1$ with randomly assigned signs) and locations $t = \eta_1$.

Motivated by \cite{jirak2014}, $\bvep_t$ is generated from the following two models:
\begin{description}
\item[(N1) ARMA($2, 2$) model.] With $\varrho_i = \varrho(i+1)^{-1}$ and $\sigma_v = 0.1\varrho^{-1}$, 
\beq
u_{j, t} &=& \sum_{i=0}^{99} \varrho_i v_{j-i, t}, \qquad v_{j, t} \sim_{\iid} \cN(0, \sigma_v^2), \label{eq:sim:u}
\\
\vep_{j, t} &=& 0.2\vep_{j, t-1}-0.3\vep_{j, t-2}+u_{j, t}+0.2u_{j, t-1}. \nonumber
\eeq
\item[(N2) Factor model.] $u_{j, t}$ is generated as in (\ref{eq:sim:u}) with $\varrho = 0.2$ and $\sigma_v = 0.5\sqrt{1-\varrho_h^2}$, and
\beqs
\vep_{j, t} &=& \varrho_h h_t + 0.2\vep_{j, t-1}-0.3\vep_{j, t-2}+u_{j, t}+0.2u_{j, t-1}, \, h_t \sim_{\iid} \cN(0, 0.1^2). \nonumber
\eeqs
\end{description}
It is easily seen that the degree of cross-sectional correlations is controlled by the choice of $\varrho \in \{0.2, 0.5\}$ in (N1)
and $\varrho_h \in \{0.5, 0.9\}$ in (N2).

\subsubsection{Results}
\label{sec:sim:single:res}

Table \ref{table:single:null} compares the Type I errors of different tests in consideration. 
Combined with the DGFM-boot procedure, the DC statistic-based tests generally manages to control the Type I errors below the nominal level ($\alpha = 0.05$) or slightly above,
with the exception of the case when $n=T=250$ and the ARMA model (N1) is used to generate the noise with $\varrho = 0.5$.
The oracle threshold for $\cTs$ leads to a very conservative test, 
which is also evident in the power performance of $\cTs$.

$\cTj$ appears to be highly sensitive and therefore vulnerable to the choice of $\wh{\sigma}_j$
particularly when the critical value is selected by the parametric bootstrap (not reported here),
since the test statistic directly depends on the largest CUSUM value attained by a single component series. 
When the noise is generated as in (N2) with $\varrho_h = 0.9$, 
the size of $\cTj$ is the closest to the nominal level,
which is attributed to the presence of a strong factor
as it leads to the scaling terms being estimated homogeneously across the panel.
This presence of a strong factor has an adverse effect on the size of $\cTeh$,
which is originally proposed for independent panel data (see also Remark \ref{rem:ef}).

To observe the power behaviour, we present the results when $n=250$ and $T=100$ as a representative example
in Tables \ref{table:single:one}--\ref{table:single:two} and report the rest in Appendix C. 
Increasing sample size generally leads to improved power and accuracy in estimated change-point location.
As for the dimensionality, its increase has different effects depending on the error generating model:
while there is no strong discernible trend with regards to increasing $n$ in the case of (N1),
it brings in dramatic improvement in the performance of $\cT^0$ for the noise generated from (N2), 
especially with increasing influence of the factor (when $\varrho_h = 0.9$).

The tests tend to lose power when the change-point is sparse, the jumps are of smaller magnitude
and its location is distanced away from the centre,
and the similar arguments apply to the location and partitioning accuracy.
Note that the Rand index occasionally decreases with increasing $m_1$,
as it measures both true positives and true negatives where the former may decrease with growing density of the change-point.
This indicates that over the range of $\delta_1$ considered here, 
some $x_{j, t}, \, j\in\Pi_1$ may not contribute to change-point detection due to small jump size $|\delta_{j, 1}|$.

When $\bvep_t$ is generated from (N1) (where the cross-sectional correlations are relatively small),
$\cTeh$ and $\cTzh$ consistently achieve high (size-corrected) power over the entire range of change-point configurations,
and the latter also achieves high location accuracy and Rand index.
In the case of (N2), $\cT^0$ outperforms all the other tests considered in our study in all change-point configurations and criteria for evaluation.
Although not reported here, the performance of $\cTzh$ in this particular setting can be improved by using a greater value for $\gamma_n$, 
which prompts the development of a data-driven way of exploiting 
the information embedded in an array of DC statistics over $\varphi \in [0, 1]$.
When the size of the data increases ($T=250$), the performance of $\cTzh$ catches up with that of $\cT^0$ in this setting. 
Comparing $\cT^0$ and $\cT^{1/2}$, the former performs superior to the latter 
except for when the change-point is highly dense and the noise is generated from (N1).

$\cTj$ performs as well as $\cT^0$ and $\cTzh$ or, occasionally, even better when the change-point is centrally located,
but its performance deteriorates greatly when $\eta_1 = 0.1T$.
This can be explained by the presence of a multiplicative factor involving $b$ in (\ref{eq:jirak}).
$\sqrt{b(T-b)/T}|\cX^j_{1, b, T}|$ amounts to a CUSUM-based change-point test 
that attains the slowest rate at which its Type II error converges to zero
among the set of CUSUM statistics studied in Section 3 of \cite{brodsky1993},
with the rate depending on $T^{\beta-1}$ instead of 
$T^{\beta/2-1/2}$ as is the case with $|\cX^j_{1, b, T}|$; see Theorem 3.5.2 of \cite{brodsky1993} for further details.
Due to the conservative behaviour of the oracle threshold as observed in Table \ref{table:single:null},
$\cTs$ does not outperform the other tests. 

Interestingly, when $\cTj$, $\cT^0$, $\cT^{1/2}$ and $\cTzh$ attain the similar level of power,
the latter two achieve higher accuracy in locating the change-point.
It is attributed to the fact that the latter usually select greater $\wh{m}^\varphi_b$ in (\ref{eq:wh:m}), 
which is evidenced by the higher Rand index values.

\begin{table}[htbp]
\caption{Type I error when $\alpha=0.05$;
$n=100$ (top) and $n=250$ (bottom).}
\label{table:single:null}
\centering
\scriptsize{
\begin{tabular}{@{}c@{}@{}c@{}|cccccc|cccccc}
\hline	\hline														
& 	& 	\multicolumn{6}{c}{$T=100$} & 						\multicolumn{6}{c}{$T=250$} 						\\	
& 	$\varrho/\varrho_h$ & 	 $\cT^0$ & 	$\cT^{1/2}$ & 	$\cTzh$ &	 $\cTj$ & 	$\cTeh$ & 	$\cTs$ &	 $\cT^0$ & 	$\cT^{1/2}$ & 	$\cTzh$ &	 $\cTj$ & 	$\cTeh$ & 	$\cTs$ 	\\	\hline
\multirow{2}{*}{(N1)} &	$0.2$ &	0.06 &	0.05 &	0.06 &	0.15 &	0.08 &	0 &	0.01 &	0.07 &	0.01 &	0.11 &	0.1 &	0	\\	
&	$0.5$ &	0.04 &	0.03 &	0.04 &	0.19 &	0.08 &	0 &	0.02 &	0.06 &	0.02 &	0.14 &	0.11 &	0	\\	
\multirow{2}{*}{(N2)} &	$0.5$ &	0.08 &	0.04 &	0.07 &	0.13 &	0.16 &	0 &	0.01 &	0.08 &	0.01 &	0.18 &	0.23 &	0	\\	
&	$0.9$ &	0.04 &	0.04 &	0.04 &	0.1 &	0.57 &	0 &	0.04 &	0.05 &	0.05 &	0.05 &	0.64 &	0	\\	\hline
\multirow{2}{*}{(N1)} &	$0.2$ &	0.06 &	0 &	0.06 &	0.15 &	0.07 &	0 &	0.07 &	0.03 &	0.07 &	0.18 &	0.09 &	0	\\	
&	$0.5$ &	0.06 &	0.01 &	0.04 &	0.21 &	0.09 &	0 &	0.12 &	0.01 &	0.1 &	0.22 &	0.05 &	0	\\	
\multirow{2}{*}{(N2)} &	$0.5$ &	0.07 &	0.04 &	0.05 &	0.04 &	0.35 &	0 &	0.02 &	0.05 &	0.05 &	0.08 &	0.34 &	0	\\	
&	$0.9$ &	0.04 &	0.07 &	0.07 &	0.1 &	0.61 &	0 &	0.05 &	0.05 &	0.05 &	0.06 &	0.75 &	0	\\	\hline
\end{tabular}}
\end{table}

\begin{sidewaystable}
\caption{$n = 250$, $T = 100$ and $\alpha=0.05$: (N1) with $\varrho = 0.2$ (top) and $\varrho = 0.5$ (bottom).}
\label{table:single:one}
\centering
\scriptsize{
\begin{tabular}{ccc|cccccc|ccccc|ccccc}
\hline\hline
& 	& 	&	\multicolumn{6}{c}{size-corrected power} & 						\multicolumn{5}{c}{location accuracy ($\%$)} & 					\multicolumn{5}{c}{Rand index} \\							
$\delta_1$ & 	$m_1$ & 	$\eta_1$ &	 $\cT^0$ & 	$\cT^{1/2}$ & 	$\cTzh$ &	 $\cTj$ & 	$\cTeh$ & 	$\cTs$ &	 $\cT^0$ & 	$\cT^{1/2}$ & 	$\cTzh$ &	 $\cTj$ & 	$\cTs$ &	 $\cT^0$ & 	$\cT^{1/2}$ & 	$\cTzh$ &	 $\cTj$ & 	$\cTs$ 	\\	\hline	
 \multirow{6}{*}{0.05} & 	\multirow{2}{*}{$\log\,n$} & 	$0.1T$ &	0.08 &	0.02 &	0.08 &	0.03 &	0.06 &	0.08 &	2 &	0 &	2 &	0 &	2 &	0.08 &	0.02 &	0.08 &	0.03 &	0.08	\\		
&	&	$0.5T$ &	0.1 &	0.03 &	0.09 &	0.11 &	0.06 &	0.08 &	0 &	0 &	0 &	6 &	0 &	0.10 &	0.03 &	0.09 &	0.11 &	0.08	\\		
&	\multirow{2}{*}{$\sqrt{n}$} & 	$0.1T$ &	0.11 &	0.03 &	0.09 &	0.07 &	0.02 &	0.07 &	1 &	0 &	1 &	0 &	1 &	0.10 &	0.03 &	0.08 &	0.06 &	0.06	\\		
&	&	$0.5T$ &	0.24 &	0.1 &	0.25 &	0.3 &	0.19 &	0.21 &	10 &	2 &	10 &	18 &	9 &	0.21 &	0.10 &	0.22 &	0.27 &	0.19	\\		
&	\multirow{2}{*}{$0.4n$} & 	$0.1T$ &	0.14 &	0.17 &	0.14 &	0.06 &	0.15 &	0.11 &	2 &	13 &	2 &	0 &	1 &	0.03 &	0.11 &	0.03 &	0.01 &	0.02	\\		
&	&	$0.5T$ &	0.56 &	1 &	0.65 &	0.77 &	0.97 &	0.6 &	23 &	96 &	31 &	53 &	26 &	0.12 &	0.77 &	0.16 &	0.18 &	0.13	\\	\hline	
 \multirow{6}{*}{0.075} & 	\multirow{2}{*}{$\log\,n$} & 	$0.1T$ &	0.09 &	0.01 &	0.09 &	0.05 &	0.04 &	0.07 &	0 &	0 &	0 &	0 &	1 &	0.09 &	0.01 &	0.09 &	0.05 &	0.07	\\		
&	&	$0.5T$ &	0.31 &	0.05 &	0.31 &	0.41 &	0.16 &	0.25 &	16 &	0 &	17 &	31 &	14 &	0.30 &	0.05 &	0.30 &	0.40 &	0.24	\\		
&	\multirow{2}{*}{$\sqrt{n}$} & 	$0.1T$ &	0.1 &	0.03 &	0.1 &	0.06 &	0.03 &	0.1 &	1 &	0 &	1 &	0 &	3 &	0.09 &	0.03 &	0.09 &	0.05 &	0.09	\\		
&	&	$0.5T$ &	0.51 &	0.15 &	0.53 &	0.69 &	0.53 &	0.48 &	22 &	5 &	23 &	47 &	22 &	0.45 &	0.15 &	0.47 &	0.62 &	0.43	\\		
&	\multirow{2}{*}{$0.4n$} & 	$0.1T$ &	0.23 &	0.87 &	0.24 &	0.11 &	0.75 &	0.24 &	12 &	86 &	12 &	0 &	16 &	0.05 &	0.66 &	0.05 &	0.02 &	0.05	\\		
&	&	$0.5T$ &	0.89 &	1 &	1 &	1 &	0.97 &	0.98 &	55 &	100 &	90 &	75 &	75 &	0.19 &	0.87 &	0.53 &	0.30 &	0.23	\\	\hline	
 \multirow{6}{*}{0.1} & 	\multirow{2}{*}{$\log\,n$} & 	$0.1T$ &	0.12 &	0.02 &	0.14 &	0.04 &	0.04 &	0.11 &	5 &	0 &	5 &	0 &	5 &	0.12 &	0.02 &	0.14 &	0.04 &	0.11	\\		
&	&	$0.5T$ &	0.55 &	0.08 &	0.55 &	0.74 &	0.34 &	0.49 &	31 &	0 &	33 &	58 &	31 &	0.53 &	0.08 &	0.53 &	0.72 &	0.48	\\		
&	\multirow{2}{*}{$\sqrt{n}$} & 	$0.1T$ &	0.15 &	0.05 &	0.15 &	0.09 &	0.09 &	0.1 &	7 &	1 &	7 &	0 &	4 &	0.13 &	0.05 &	0.13 &	0.08 &	0.09	\\		
&	&	$0.5T$ &	0.76 &	0.47 &	0.82 &	0.94 &	0.96 &	0.74 &	52 &	33 &	58 &	74 &	57 &	0.68 &	0.47 &	0.74 &	0.86 &	0.66	\\		
&	\multirow{2}{*}{$0.4n$} & 	$0.1T$ &	0.42 &	1 &	0.68 &	0.11 &	0.97 &	0.51 &	28 &	100 &	55 &	0 &	34 &	0.09 &	0.84 &	0.26 &	0.02 &	0.11	\\		
&	&	$0.5T$ &	1 &	1 &	1 &	1 &	0.97 &	1 &	76 &	100 &	99 &	85 &	97 &	0.21 &	0.93 &	0.69 &	0.47 &	0.29	\\	\hline	\hline
 \multirow{6}{*}{0.05} & 	\multirow{2}{*}{$\log\,n$} & 	$0.1T$ &	0.08 &	0.02 &	0.08 &	0.03 &	0.06 &	0.08 &	2 &	0 &	2 &	0 &	2 &	0.08 &	0.02 &	0.08 &	0.03 &	0.08	\\		
&	&	$0.5T$ &	0.1 &	0.03 &	0.09 &	0.11 &	0.06 &	0.08 &	0 &	0 &	0 &	6 &	0 &	0.10 &	0.03 &	0.09 &	0.11 &	0.08	\\		
&	\multirow{2}{*}{$\sqrt{n}$} & 	$0.1T$ &	0.11 &	0.03 &	0.09 &	0.07 &	0.02 &	0.07 &	1 &	0 &	1 &	0 &	1 &	0.10 &	0.03 &	0.08 &	0.06 &	0.06	\\		
&	&	$0.5T$ &	0.24 &	0.1 &	0.25 &	0.3 &	0.19 &	0.21 &	10 &	2 &	10 &	18 &	9 &	0.21 &	0.10 &	0.22 &	0.27 &	0.19	\\		
&	\multirow{2}{*}{$0.4n$} & 	$0.1T$ &	0.14 &	0.17 &	0.14 &	0.06 &	0.15 &	0.11 &	2 &	13 &	2 &	0 &	1 &	0.03 &	0.11 &	0.03 &	0.01 &	0.02	\\		
&	&	$0.5T$ &	0.56 &	1 &	0.65 &	0.77 &	0.97 &	0.6 &	23 &	96 &	31 &	53 &	26 &	0.12 &	0.77 &	0.16 &	0.18 &	0.13	\\	\hline	
 \multirow{6}{*}{0.075} & 	\multirow{2}{*}{$\log\,n$} & 	$0.1T$ &	0.09 &	0.01 &	0.09 &	0.05 &	0.04 &	0.07 &	0 &	0 &	0 &	0 &	1 &	0.09 &	0.01 &	0.09 &	0.05 &	0.07	\\		
&	&	$0.5T$ &	0.31 &	0.05 &	0.31 &	0.41 &	0.16 &	0.25 &	16 &	0 &	17 &	31 &	14 &	0.30 &	0.05 &	0.30 &	0.40 &	0.24	\\		
&	\multirow{2}{*}{$\sqrt{n}$} & 	$0.1T$ &	0.1 &	0.03 &	0.1 &	0.06 &	0.03 &	0.1 &	1 &	0 &	1 &	0 &	3 &	0.09 &	0.03 &	0.09 &	0.05 &	0.09	\\		
&	&	$0.5T$ &	0.51 &	0.15 &	0.53 &	0.69 &	0.53 &	0.48 &	22 &	5 &	23 &	47 &	22 &	0.45 &	0.15 &	0.47 &	0.62 &	0.43	\\		
&	\multirow{2}{*}{$0.4n$} & 	$0.1T$ &	0.23 &	0.87 &	0.24 &	0.11 &	0.75 &	0.24 &	12 &	86 &	12 &	0 &	16 &	0.05 &	0.66 &	0.05 &	0.02 &	0.05	\\		
&	&	$0.5T$ &	0.89 &	1 &	1 &	1 &	0.97 &	0.98 &	55 &	100 &	90 &	75 &	75 &	0.19 &	0.87 &	0.53 &	0.30 &	0.23	\\	\hline	
 \multirow{6}{*}{0.1} & 	\multirow{2}{*}{$\log\,n$} & 	$0.1T$ &	0.12 &	0.02 &	0.14 &	0.04 &	0.04 &	0.11 &	5 &	0 &	5 &	0 &	5 &	0.12 &	0.02 &	0.14 &	0.04 &	0.11	\\		
&	&	$0.5T$ &	0.55 &	0.08 &	0.55 &	0.74 &	0.34 &	0.49 &	31 &	0 &	33 &	58 &	31 &	0.53 &	0.08 &	0.53 &	0.72 &	0.48	\\		
&	\multirow{2}{*}{$\sqrt{n}$} & 	$0.1T$ &	0.15 &	0.05 &	0.15 &	0.09 &	0.09 &	0.1 &	7 &	1 &	7 &	0 &	4 &	0.13 &	0.05 &	0.13 &	0.08 &	0.09	\\		
&	&	$0.5T$ &	0.76 &	0.47 &	0.82 &	0.94 &	0.96 &	0.74 &	52 &	33 &	58 &	74 &	57 &	0.68 &	0.47 &	0.74 &	0.86 &	0.66	\\		
&	\multirow{2}{*}{$0.4n$} & 	$0.1T$ &	0.42 &	1 &	0.68 &	0.11 &	0.97 &	0.51 &	28 &	100 &	55 &	0 &	34 &	0.09 &	0.84 &	0.26 &	0.02 &	0.11	\\		
&	&	$0.5T$ &	1 &	1 &	1 &	1 &	0.97 &	1 &	76 &	100 &	99 &	85 &	97 &	0.21 &	0.93 &	0.69 &	0.47 &	0.29	\\	\hline	
&	&	average &	0.35 &	0.34 &	0.38 &	0.37 &	0.41 &	0.35 &	19.06 &	29.78 &	24.78 &	24.83 &	22.11 &	0.20 &	0.29 &	0.26 &	0.24 &	0.19	\\	\hline	
\end{tabular}}
\end{sidewaystable}

\begin{sidewaystable}
\caption{$n = 250$, $T = 100$ and $\alpha=0.05$: (N2) with $\varrho_h = 0.5$ (top) and $\varrho_h = 0.9$ (bottom).}
\label{table:single:two}
\centering
\scriptsize{
\begin{tabular}{ccc|cccccc|ccccc|ccccc}
\hline\hline
& 	& 	&	\multicolumn{6}{c}{size-corrected power} & 						\multicolumn{5}{c}{location accuracy ($\%$)} & 					\multicolumn{5}{c}{Rand index} \\							
$\delta_1$ & 	$m_1$ & 	$\eta_1$ &	 $\cT^0$ & 	$\cT^{1/2}$ & 	$\cTzh$ &	 $\cTj$ & 	$\cTeh$ & 	$\cTs$ &	 $\cT^0$ & 	$\cT^{1/2}$ & 	$\cTzh$ &	 $\cTj$ & 	$\cTs$ &	 $\cT^0$ & 	$\cT^{1/2}$ & 	$\cTzh$ &	 $\cTj$ & 	$\cTs$ 	\\	\hline	
  \multirow{6}{*}{0.05} & 	\multirow{2}{*}{$\log\,n$} & 	$0.1T$ &	0.22 &	0.06 &	0.14 &	0.04 &	0 &	0.02 &	4 &	1 &	1 &	0 &	0 &	0.21 &	0.06 &	0.14 &	0.04 &	0.02	\\	
&	&	$0.5T$ &	0.3 &	0.07 &	0.19 &	0.08 &	0 &	0.05 &	12 &	0 &	6 &	6 &	1 &	0.29 &	0.07 &	0.18 &	0.08 &	0.05	\\	
&	\multirow{2}{*}{$\sqrt{n}$} & 	$0.1T$ &	0.2 &	0.1 &	0.16 &	0.02 &	0 &	0.02 &	5 &	1 &	4 &	0 &	0 &	0.18 &	0.10 &	0.15 &	0.02 &	0.02	\\	
&	&	$0.5T$ &	0.45 &	0.1 &	0.21 &	0.09 &	0 &	0.03 &	18 &	1 &	6 &	4 &	0 &	0.40 &	0.10 &	0.19 &	0.08 &	0.03	\\	
&	\multirow{2}{*}{$0.4n$} & 	$0.1T$ &	0.4 &	0.08 &	0.25 &	0.07 &	0 &	0.02 &	20 &	2 &	11 &	0 &	0 &	0.08 &	0.05 &	0.10 &	0.01 &	0.00	\\	
&	&	$0.5T$ &	0.84 &	0.49 &	0.88 &	0.55 &	0.62 &	0.1 &	40 &	39 &	67 &	39 &	6 &	0.18 &	0.38 &	0.45 &	0.14 &	0.02	\\	\hline
 \multirow{6}{*}{0.075} & 	\multirow{2}{*}{$\log\,n$} & 	$0.1T$ &	0.24 &	0.04 &	0.16 &	0.02 &	0 &	0.03 &	7 &	0 &	6 &	0 &	0 &	0.23 &	0.04 &	0.16 &	0.02 &	0.03	\\	
&	&	$0.5T$ &	0.68 &	0.08 &	0.48 &	0.26 &	0 &	0.11 &	40 &	1 &	27 &	16 &	8 &	0.66 &	0.08 &	0.47 &	0.25 &	0.11	\\	
&	\multirow{2}{*}{$\sqrt{n}$} & 	$0.1T$ &	0.32 &	0.09 &	0.21 &	0.07 &	0 &	0.03 &	20 &	1 &	8 &	0 &	1 &	0.28 &	0.09 &	0.19 &	0.06 &	0.03	\\	
&	&	$0.5T$ &	0.91 &	0.13 &	0.76 &	0.55 &	0.05 &	0.17 &	53 &	2 &	54 &	41 &	10 &	0.81 &	0.13 &	0.69 &	0.50 &	0.15	\\	
&	\multirow{2}{*}{$0.4n$} & 	$0.1T$ &	0.7 &	0.26 &	0.61 &	0.07 &	0.33 &	0.03 &	44 &	19 &	48 &	0 &	2 &	0.15 &	0.19 &	0.28 &	0.02 &	0.01	\\	
&	&	$0.5T$ &	1 &	1 &	1 &	0.9 &	0.7 &	0.7 &	54 &	100 &	100 &	65 &	50 &	0.21 &	0.90 &	0.72 &	0.31 &	0.15	\\	\hline
 \multirow{6}{*}{0.1} & 	\multirow{2}{*}{$\log\,n$} & 	$0.1T$ &	0.37 &	0.05 &	0.21 &	0.06 &	0 &	0.03 &	18 &	1 &	11 &	0 &	0 &	0.36 &	0.05 &	0.20 &	0.06 &	0.03	\\	
&	&	$0.5T$ &	0.95 &	0.08 &	0.85 &	0.68 &	0 &	0.49 &	71 &	0 &	64 &	55 &	41 &	0.92 &	0.08 &	0.83 &	0.67 &	0.48	\\	
&	\multirow{2}{*}{$\sqrt{n}$} & 	$0.1T$ &	0.66 &	0.07 &	0.35 &	0.03 &	0 &	0.04 &	46 &	2 &	23 &	0 &	2 &	0.59 &	0.07 &	0.32 &	0.03 &	0.04	\\	
&	&	$0.5T$ &	0.99 &	0.15 &	0.98 &	0.88 &	0.69 &	0.7 &	80 &	3 &	87 &	74 &	60 &	0.88 &	0.15 &	0.90 &	0.82 &	0.63	\\	
&	\multirow{2}{*}{$0.4n$} & 	$0.1T$ &	0.92 &	0.86 &	0.99 &	0.04 &	0.7 &	0.2 &	77 &	84 &	96 &	1 &	17 &	0.19 &	0.73 &	0.63 &	0.01 &	0.04	\\	
&	&	$0.5T$ &	1 &	1 &	1 &	0.97 &	0.7 &	0.99 &	77 &	100 &	100 &	86 &	96 &	0.21 &	0.94 &	0.84 &	0.52 &	0.27	\\	\hline
 \multirow{6}{*}{0.05} & 	\multirow{2}{*}{$\log\,n$} & 	$0.1T$ &	0.18 &	0.05 &	0.05 &	0.05 &	0 &	0.02 &	12 &	2 &	2 &	0 &	0 &	0.17 &	0.05 &	0.05 &	0.05 &	0.02	\\	
&	&	$0.5T$ &	0.75 &	0.05 &	0.05 &	0.24 &	0 &	0.03 &	43 &	0 &	0 &	12 &	2 &	0.73 &	0.05 &	0.05 &	0.23 &	0.03	\\	
&	\multirow{2}{*}{$\sqrt{n}$} & 	$0.1T$ &	0.33 &	0.06 &	0.06 &	0.05 &	0 &	0.03 &	21 &	1 &	1 &	0 &	1 &	0.29 &	0.06 &	0.06 &	0.04 &	0.03	\\	
&	&	$0.5T$ &	0.96 &	0.07 &	0.07 &	0.39 &	0 &	0.1 &	69 &	0 &	0 &	28 &	7 &	0.85 &	0.07 &	0.07 &	0.35 &	0.09	\\	
&	\multirow{2}{*}{$0.4n$} & 	$0.1T$ &	0.69 &	0.07 &	0.07 &	0.05 &	0 &	0.03 &	47 &	2 &	2 &	0 &	1 &	0.14 &	0.06 &	0.06 &	0.01 &	0.01	\\	
&	&	$0.5T$ &	1 &	0.12 &	0.17 &	0.74 &	0.39 &	0.21 &	60 &	2 &	7 &	47 &	13 &	0.21 &	0.08 &	0.11 &	0.24 &	0.05	\\	\hline
 \multirow{6}{*}{0.075} & 	\multirow{2}{*}{$\log\,n$} & 	$0.1T$ &	0.6 &	0.05 &	0.05 &	0.04 &	0 &	0.04 &	43 &	2 &	2 &	0 &	2 &	0.58 &	0.05 &	0.05 &	0.04 &	0.04	\\	
&	&	$0.5T$ &	1 &	0.07 &	0.09 &	0.76 &	0 &	0.27 &	83 &	0 &	2 &	64 &	22 &	0.97 &	0.07 &	0.09 &	0.74 &	0.26	\\	
&	\multirow{2}{*}{$\sqrt{n}$} & 	$0.1T$ &	0.89 &	0.09 &	0.09 &	0.06 &	0 &	0.05 &	76 &	2 &	2 &	0 &	2 &	0.79 &	0.09 &	0.09 &	0.05 &	0.04	\\	
&	&	$0.5T$ &	1 &	0.09 &	0.11 &	0.85 &	0.28 &	0.43 &	88 &	0 &	5 &	70 &	37 &	0.89 &	0.09 &	0.10 &	0.79 &	0.39	\\	
&	\multirow{2}{*}{$0.4n$} & 	$0.1T$ &	0.94 &	0.09 &	0.09 &	0.12 &	0.25 &	0.12 &	70 &	2 &	2 &	1 &	6 &	0.20 &	0.07 &	0.07 &	0.03 &	0.03	\\	
&	&	$0.5T$ &	1 &	0.73 &	0.92 &	0.99 &	0.44 &	0.77 &	81 &	64 &	87 &	79 &	59 &	0.21 &	0.65 &	0.75 &	0.51 &	0.22	\\	\hline
 \multirow{6}{*}{0.1} & 	\multirow{2}{*}{$\log\,n$} & 	$0.1T$ &	0.9 &	0.08 &	0.08 &	0.05 &	0 &	0.06 &	86 &	2 &	2 &	1 &	4 &	0.87 &	0.08 &	0.08 &	0.05 &	0.06	\\	
&	&	$0.5T$ &	1 &	0.06 &	0.25 &	0.97 &	0.14 &	0.8 &	92 &	0 &	21 &	85 &	73 &	0.97 &	0.06 &	0.24 &	0.96 &	0.78	\\	
&	\multirow{2}{*}{$\sqrt{n}$} & 	$0.1T$ &	1 &	0.06 &	0.06 &	0.06 &	0 &	0.14 &	95 &	2 &	2 &	2 &	9 &	0.89 &	0.06 &	0.06 &	0.05 &	0.13	\\	
&	&	$0.5T$ &	1 &	0.07 &	0.59 &	1 &	0.44 &	0.96 &	90 &	0 &	57 &	90 &	89 &	0.89 &	0.07 &	0.55 &	0.97 &	0.88	\\	
&	\multirow{2}{*}{$0.4n$} & 	$0.1T$ &	0.99 &	0.1 &	0.29 &	0.16 &	0.44 &	0.23 &	85 &	6 &	25 &	3 &	17 &	0.21 &	0.08 &	0.21 &	0.04 &	0.06	\\	
&	&	$0.5T$ &	1 &	1 &	1 &	1 &	0.44 &	0.99 &	86 &	99 &	99 &	93 &	92 &	0.21 &	0.96 &	0.90 &	0.80 &	0.45	\\	\hline
&	&	average & 	0.73 &	0.21 &	0.38 &	0.36 &	0.18 &	0.25 &	53.14 &	15.08 &	28.81 &	26.72 &	20.28 &	0.47 &	0.19 &	0.31 &	0.27 &	0.16	\\	\hline
\end{tabular}}
\end{sidewaystable}

\subsection{Multiple change-point detection}
\label{sec:sim:multiple}

In this section, we evaluate the empirical behaviour 
of the DCBS algorithm with $\cT^0$ and $\cTzh$,
based on their good performance observed in Section \ref{sec:sim:single:res}. 
For comparison, we also investigate the performance of the SBS algorithm 
furnished with the oracle threshold.

Fixing the dimensionality and the size of data at $n = 250$ and $T = 250$,
we consider piecewise constant signal $\{f_{j, t}\}_{t=1}^{T}$ containing three change-points as follows:
at $t = \eta_r$, an index set $\Pi_r$ of cardinality $m_r$ is randomly drawn from $\{1, \ldots, n\}$,
where $|\delta_{j, r}| \sim_{\iid} \cU(0.75\delta_r, 1.25\delta_r)$ for $j \in \Pi_r$.
We set $(\eta_1, m_1, \delta_1) = ([0.3T], [0.75n], 0.050)$, $(\eta_2, m_2, \delta_2) = ([0.6T], [0.25n], 0.087)$, 
and $(\eta_3, m_3, \delta_3) = ([0.8T], [0.1n], 0.140)$ such that $m_r\delta_r^2$ remains identical over all $r=1, 2, 3$.
The noise $\bvep_t$ is generated as in (N1) and (N2) of Section \ref{sec:sim:single:model} with varying $\varrho$ and $\varrho_h$.
We set $B = 100$, $d_T = 5$ and $L_T = [\log_2(\log\,T+1)]$ where the latter is chosen to permit
a growing number ($\log\,T$) of change-points in the data.
To account for multiple testing, the Bonferroni's correction is adopted by setting $\alpha = \alpha^*/(2^{L_T}-1)$ with $\alpha^* = 0.05$.
We report the total number of estimated change-points ($\wh{N}$, in $\%$) and their location accuracy ($|\heta_r - \eta_r| < \log\,T$, in $\%$),
over $100$ simulated sample paths for each setting in Table \ref{table:multi} and Figures \ref{fig:multi:one}--\ref{fig:multi:two}.

Overall, the BS algorithm applied with $\cTzh$ performs the best 
in detecting all three change-points as well as identifying their locations over $80\%$ of the simulated data.
In comparison, $\cT^0$ or $\cTs$ tend to miss $\eta_1$ which is associated with the smallest jump size,
while all three methods tend to detect $\eta_3$ the best, which is associated with the largest jump size. 
The behaviour of $\cT^0$ and $\cTzh$ change dramatically when $\bvep_t$ is generated
from (N2) with a strong factor ($\varrho_h = 0.9$),
where the performance of $\cT^0$ improves (identifying all three change-points over $70\%$ of the simulated data) 
whereas that of the latter deteriorates greatly.
As observed in Section \ref{sec:sim:single:res}, using a larger scaling term $\gamma_n$ may improve the performance of $\cTzh$ in this setting.
While the performance of $\cTs$ is better when $\bvep_t$ is generated from (N2) rather than (N1),
the choice of threshold for $\cTs$ turns out to be too conservative overall.
Although each $\vep_{j, t}$ is generated from the identical model, 
employing $n$ thresholds for the $n$-dimensional panel data may lead to better small sample performance of $\cTs$.

\begin{table}[htbp]
\caption{Summary of the total number of estimated change-points
and their location accuracy: $n = 250$, $T = 250$ and $\alpha^*=0.05$.}
\label{table:multi}
\centering
\scriptsize{
\begin{tabular}{c|c|c|cccccc|ccc}
\hline\hline
&	&	&	\multicolumn{6}{c}{$\wh{N}$ (\%)} &						\multicolumn{3}{c}{accuracy (\%)} 			\\	
&	$\varrho/\varrho_h$ &	&	0 &	1 &	2 &	\textbf{3} &	4 &	$\ge$5 &	$\eta_1$ &	$\eta_2$ &	$\eta_3$ 	\\	\hline
\multirow{6}{*}{(N1)} &	\multirow{3}{*}{0.2} &	$\cT^0$ &	0 &	9 &	37 &	53 &	0 &	1 &	35 &	71 &	89	\\	
&	&	$\cTzh$ &	0 &	1 &	12 &	85 &	1 &	1 &	90 &	87 &	93	\\	
&	&	$\cTs$ &	8 &	66 &	26 &	0 &	0 &	0 &	0 &	37 &	77	\\	\cline{2-12}
&	\multirow{3}{*}{0.5} &	$\cT^0$ &	0 &	4 &	37 &	56 &	3 &	0 &	46 &	71 &	92	\\	
&	&	$\cTzh$ &	0 &	1 &	8 &	88 &	3 &	0 &	95 &	89 &	96	\\	
&	&	$\cTs$ &	4 &	67 &	29 &	0 &	0 &	0 &	0 &	46 &	75	\\	\hline
\multirow{6}{*}{(N2)} &	\multirow{3}{*}{0.5} &	$\cT^0$ &	1 &	3 &	31 &	65 &	0 &	0 &	34 &	81 &	97	\\	
&	&	$\cTzh$ &	0 &	2 &	7 &	83 &	8 &	0 &	96 &	95 &	97	\\	
&	&	$\cTs$ &	2 &	21 &	76 &	1 &	0 &	0 &	2 &	83 &	90	\\	\cline{2-12}
&	\multirow{3}{*}{0.9} &	$\cT^0$ &	0 &	2 &	10 &	87 &	1 &	0 &	74 &	97 &	98	\\	
&	&	$\cTzh$ &	3 &	35 &	37 &	24 &	1 &	0 &	44 &	38 &	92	\\	
&	&	$\cTs$ &	0 &	22 &	69 &	9 &	0 &	0 &	9 &	76 &	100	\\	\hline
\end{tabular}}
\end{table}

\begin{figure}[htbp]
\centering
\begin{subfigure}{.475\textwidth}
\centering
\includegraphics[width=1\linewidth]{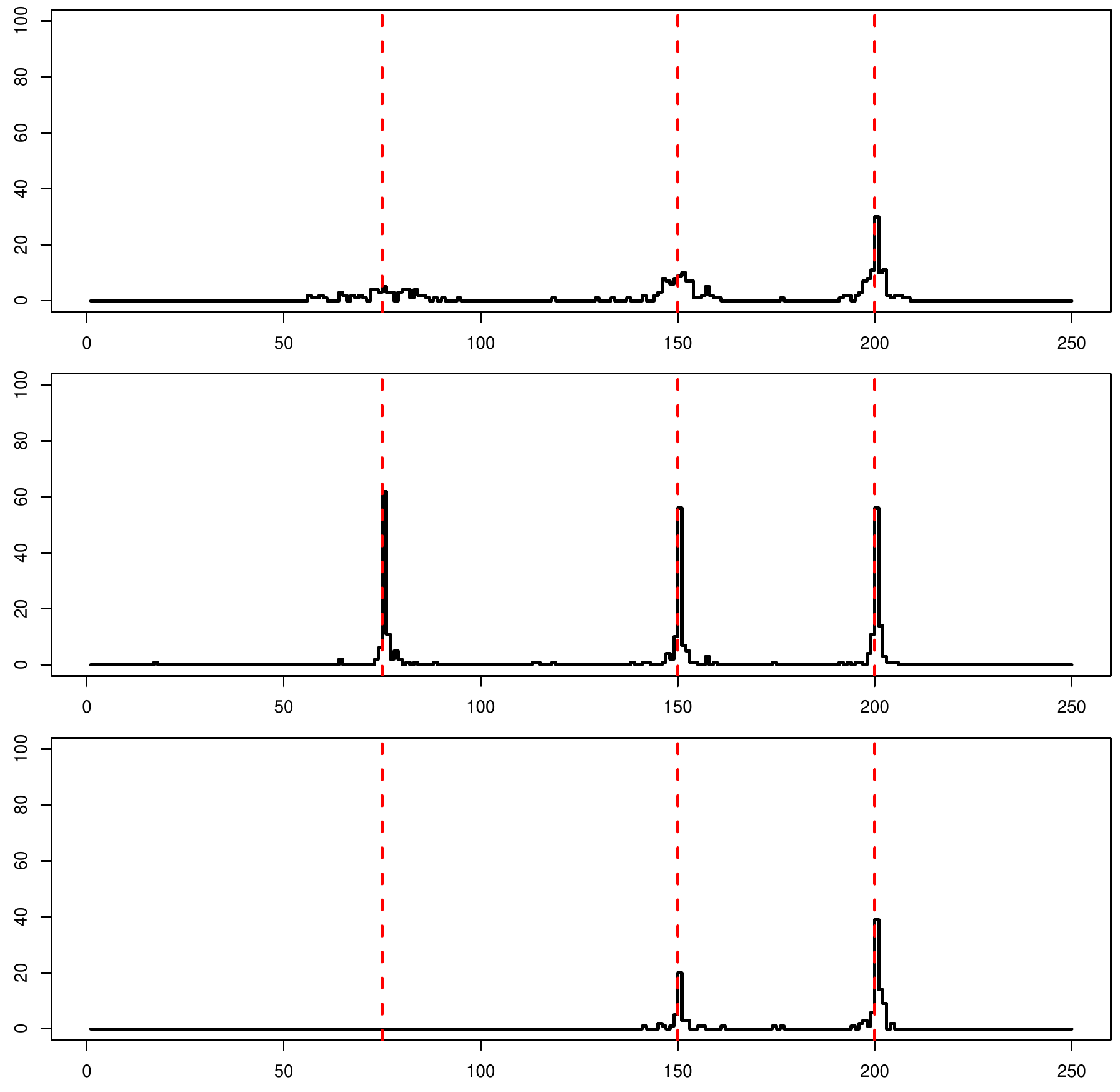}  
\end{subfigure}
\begin{subfigure}{.475\textwidth}
\centering
\includegraphics[width=1\linewidth]{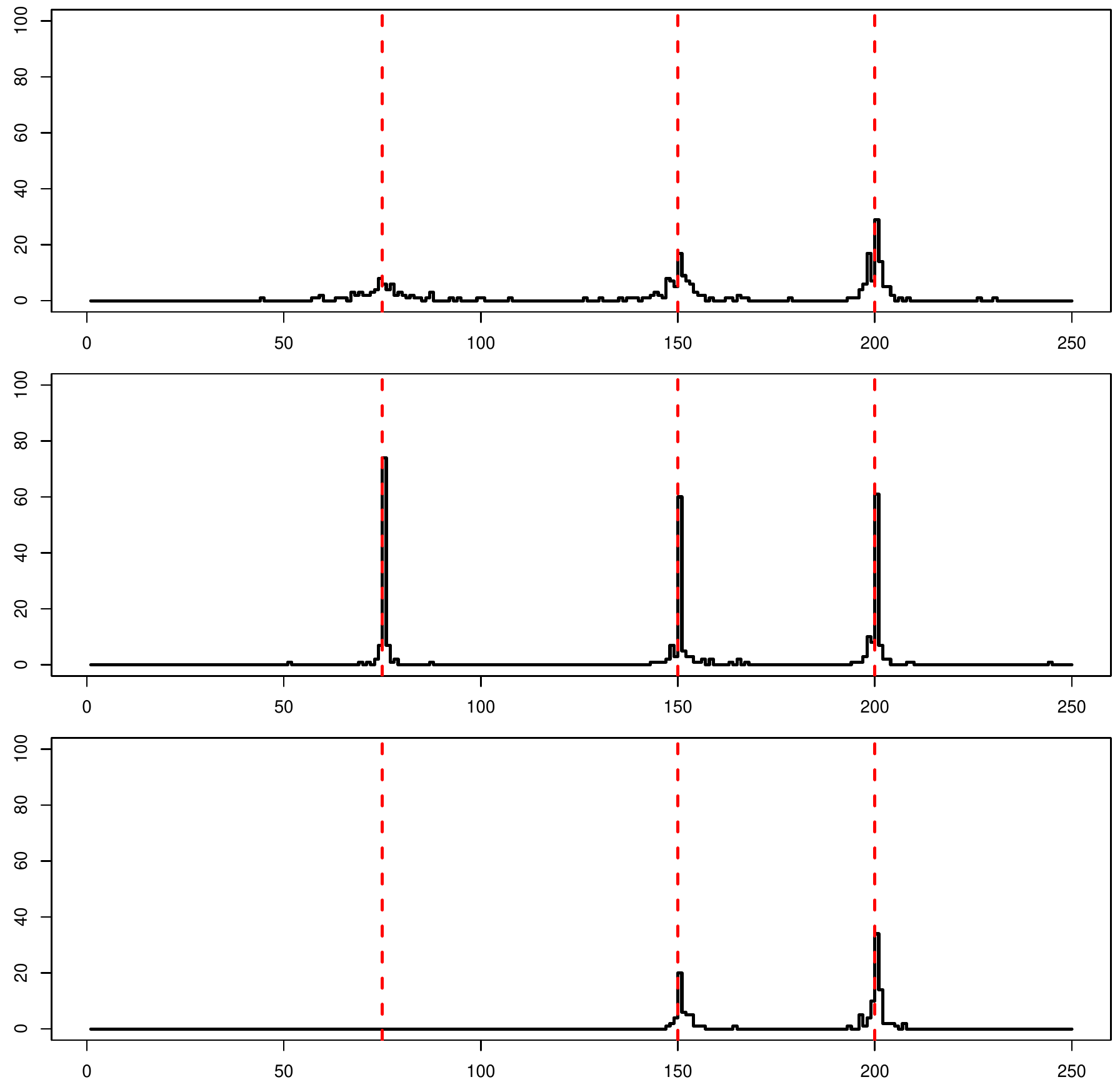}
\end{subfigure}
\caption{(N1) with $\varrho = 0.2$ (left) and $\varrho = 0.5$ (right): the locations of the change-points detected by the BS algorithm in combination with $\cT^0$ (top), $\cTzh$ (middle) and $\cTs$ (bottom); vertical lines indicate the locations of true $\eta_r, \, r=1, 2, 3$.}
\label{fig:multi:one}
\end{figure}

\begin{figure}[htbp]
\centering
\begin{subfigure}{.475\textwidth}
\centering
\includegraphics[width=1\linewidth]{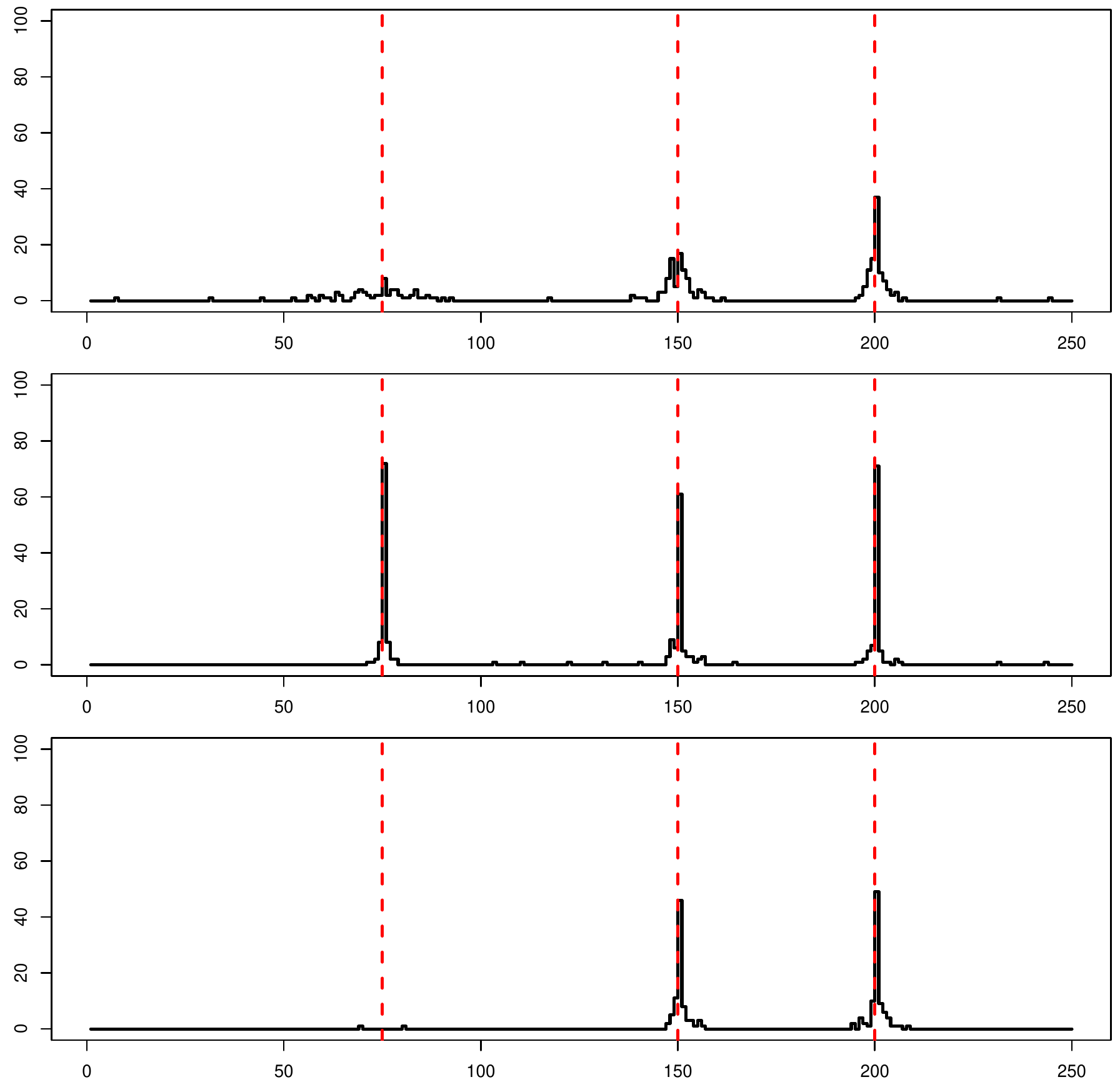}  
\end{subfigure}
\begin{subfigure}{.475\textwidth}
\centering
\includegraphics[width=1\linewidth]{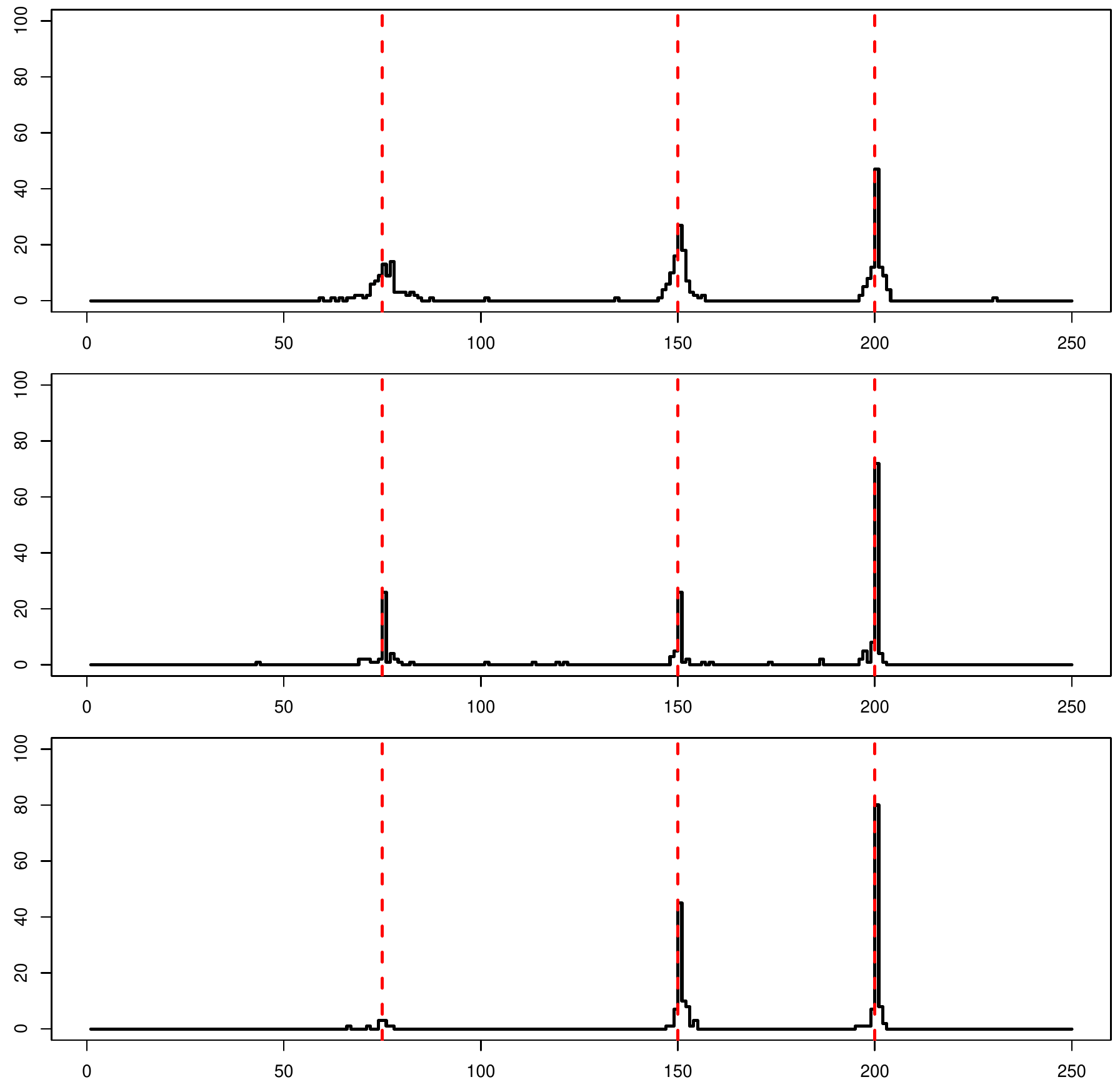}
\end{subfigure}
\caption{(N2) with $\varrho_h = 0.5$ (left) and $\varrho_h = 0.9$ (right).}
\label{fig:multi:two}
\end{figure}

\section{Application to financial time series data}
\label{sec:real}

We analysed the log returns of the closing prices of all S\&P $100$ component stocks 
between February $24$, $2015$ and February $23$, $2016$,
which are denoted by $y_{i, t}, \, i=1, \ldots, \wt{n}; \, t=1, \ldots, T$
with $\wt{n} = 88$ (only those components which remained in the index for a longer period were included) and $T=252$.
\cite{jirak2014} analysed a similar financial dataset for a single change in its mean and variance, respectively.
However, considering that 
(i) log returns are often modelled to have zero-mean and time-varying conditional variance using conditionally heteroscedastic models,
and (ii) it is difficult to rule out the possible existence of multiple change-points,
we chose to perform the change-point analysis in the (unconditional) second-order structure of $y_{i, t}$ using the DCBS algorithm.
For the purpose, wavelet-based periodogram and cross-periodogram sequences of $y_{i, t}$ were computed,
which were also adopted to comprise the input panel data to the SBS-MVTS algorithm in \cite{cho2015}.
Any change-point in the autocovariance and cross-covariance structure of $y_{i, t}$ is detectable
from examining the  wavelet (cross-)periodogram sequences;
for further details, see Section 3.1 of \cite{cho2015}.

We used Haar wavelets at the two finest scales to produce the periodogram sequences,
which are denoted by $x_{j, t}, \, j=1, \ldots, n = \wt{n}(\wt{n}+1); \, t=1, \ldots, T$.
The DCBS algorithm with $\cDzh_m$ detects a single change-point at $t = 220$,
which corresponds to January $6$, $2016$. 
It has been noted that the first week of trading in $2016$ marked the worst five-day start to a year ever,
according to S\&P Dow Jones Indices (Financial Times, \url{http://www.ft.com/fastft/2016/01/07/sp-500-logs-worst-annual-kick-off-on-record/}).
For example, the S\&P $500$ index dropped by $4.9\%$ during the period,
and the Dow Jones Industrial Average by $6.19\%$.
Figure \ref{fig:real} shows $y_{i, t}$ (left) and the pointwise maximum of the DC statistics at the first iteration of the DCBS algorithm (right),
where such behaviour of the financial market at the beginning of $2016$ is reflected as a large peak in the latter.

\begin{figure}[htbp]
\centering
\includegraphics[width=1\linewidth]{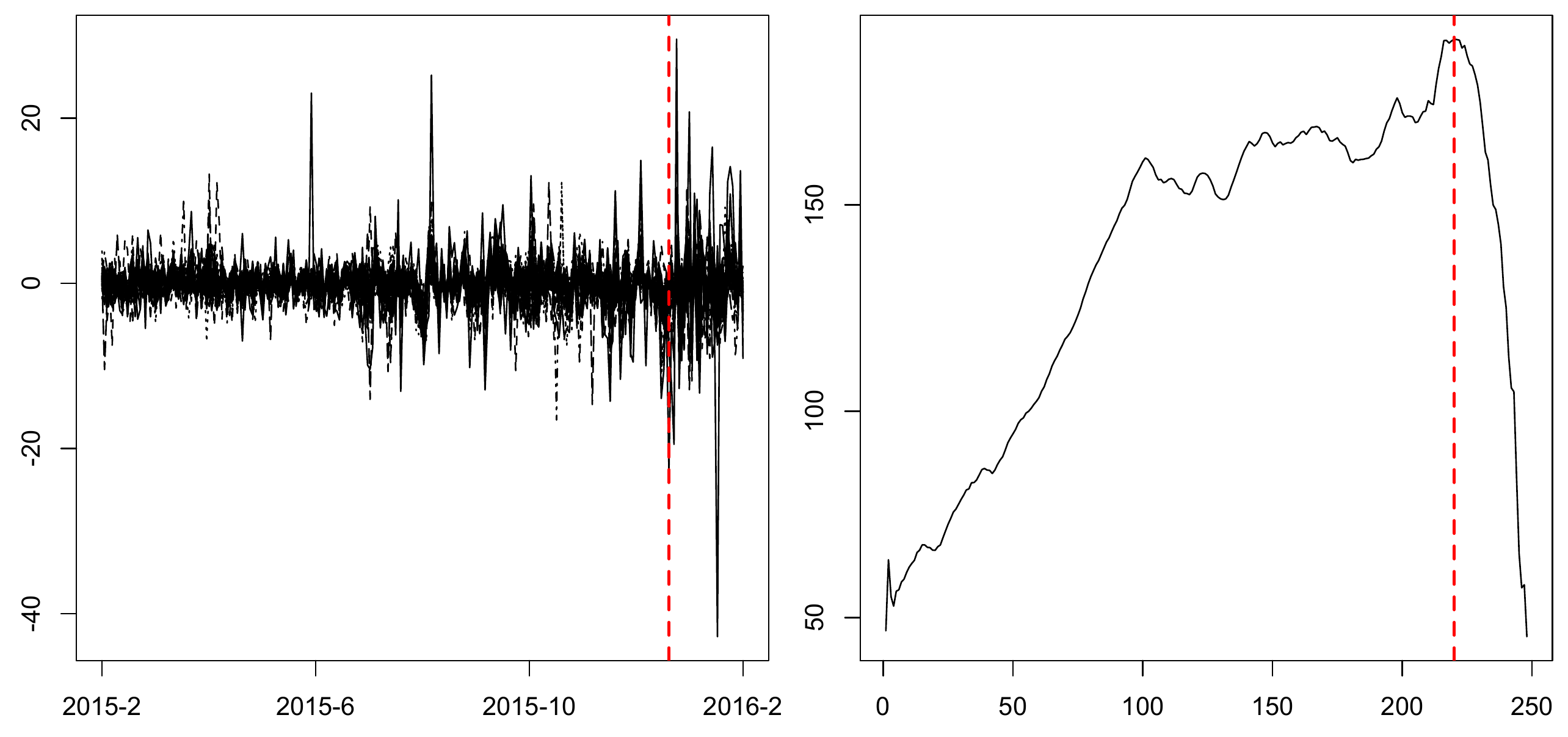}  
\caption{Log returns of S\&P $100$ component stock prices ($y_{i, t}$) between February $24$, $2015$ and February $23$, $2016$ (left);
pointwise maximum of $\cDzh_m(\{|\cX^{(j)}_{1, b, T}|\}_{j=1}^{n})$ over $b=1, \ldots, T-1$ (right);
the vertical broken line denotes the location of the estimated change-point.}
\label{fig:real}
\end{figure}

\section{Conclusions}
\label{sec:conc}

In this paper, we have proposed the DC statistic, 
a novel way of aggregating high-dimensional CUSUM series across the panel for change-point analysis,
and showed its consistency in single and multiple change-point detection both theoretically and empirically.
We conclude by listing possible future research projects. 

\begin{itemize}
\item The DC statistic can be applied to detect changes in high-dimensional time series 
besides those in the mean of panel data.
For example, the DCBS algorithm is easily extended to detect change-points in 
both autocovariance and cross-covariance structure of $n$-dimensional time series,
by taking panel data consisting of local periodogram and cross-periodogram sequences as an input, 
see the real data analysis in Section \ref{sec:real}.

Moreover, the DC operator may be regarded as a generic tool that can be adopted to aggregate
multiple series of statistics cross-sectionally, 
the result of which can be utilised for panel data analysis beyond change-point detection.

\item Empirically, it was shown that $\cTzh$ generally outperforms $\cT^\varphi$ for any $\varphi \in \{0, 1/2\}$,
while in the presence of strong cross-correlations, $\cT^0$ performs better than the rest.
This opens up several possibilities for future research, including
the investigation of the ``optimal'' way of exploiting the information contained in 
the infinite array of DC statistics over $\varphi \in [0, 1]$. 

\item Proposed here for test criterion selection, the GDFM-boot algorithm showed good empirical performance,
but it remains to investigate the validity of its application to change-point analysis
by establishing that the bootstrap scheme mimics the correct second-order structure 
for a large class of time series processes. 
The GDFM-boot will have applicability to a wide range of inference and estimation problems involving high-dimensional time series
beyond the context of change-point analysis. For example, a key task in factor analysis is to estimate the number of common factors 
that drive the pervasive cross-sectional correlations, and several information criterion-type estimators have been proposed.
However, there is lack of any attempt at statistical inference on the factor number, 
e.g., by constructing its confidence interval, and the GDFM-boot can be adopted for such tasks. 
\end{itemize}


\section{Proofs}
\label{sec:pf}

\subsection{Preliminary results}
\label{sec:pf:pre}

We first prove a set of lemmas that are stepping stones for the proofs of Theorems \ref{thm:zero}--\ref{thm:four}.
We assume that (A1)--(A3) and (A6) hold in all the lemmas where applicable,
and the notations $C_i, \, c_i, \ i=0, 1, \ldots$ are adopted
to denote fixed positive constants throughout the proofs.
Further, $\wt\vep_{j, t}$ denotes the scaled $\vep_{j, t}$ with respect to an estimator $\wh\sigma_j$ satisfying (A6),
and $\wt{x}_{j, t}$ and $\wt{f}_{j, t}$ are defined similarly.
Throughout, we operate within $\cE_\sigma$ defined in (A6).
\vspace{5pt}
\begin{lem}
\label{lem:one}
Defining the set
$\cI_1 = \{(s, e):\, 1 \le s < e \le T,  \  e-s+1 > d_T = [C\log^2\,T]\}$
and the event
$\cE_1 = \{\max_{1 \le j \le n}\max_{(s, e) \in \cI_1}(e-s+1)^{-1/2}\vert\sum_{t=s}^e \wt\vep_{j, t} \vert \le \log\,T\}$
with some fixed constant $C>0$, we have $\pr(\cE_1) \to 1$ as $T \to \infty$.
\end{lem}
\begin{proof}
We first study the following probability
\begin{eqnarray}
\label{lem:one:event}
\pr\left(\frac{1}{\sqrt{e-s+1}}\left\vert\sum_{t=s}^e \wt\vep_{j, t} \right\vert > \log\,T\right).
\end{eqnarray}
Let $d = e-s+1$. Theorem 1.4 of \cite{bosq1998} showed that under the conditions imposed in (A1), 
the probability in (\ref{lem:one:event}) is bounded from the above by
\begin{align}
& \left\{\frac{2d}{q} + 2\left(1+\frac{\log^2\,T}{25d\E(\wt\vep_{j, t}^2)+5C_{\vep}\sqrt{d}\log\,T}\right)\right\}
\exp\left(-\frac{q\log^2\,T}{25d\E(\wt\vep_{j, t}^2)+5C_{\vep}\sqrt{d}\log\,T}\right)
\nonumber
\\
& + 11C_\alpha d\left[1+\frac{5\{\E(\wt\vep_{j, t}^2)\}^{2/5}\sqrt{d}}{\log\,T}\right](\mu^{[d/(q+1)]})^{4/5}
\label{lem:one:eq}
\end{align}
where $q \in \{1, \ldots, [d/2]\}$.
With the choice $q = [c_qd/\log\,T]$ for some $c_q > 0$, (\ref{lem:one:eq}) is bounded by
\begin{flalign*}
& 2\left(\frac{\log\,T}{c_q} + 2\right)\exp\left\{-\frac{c_q\log\,T}{25\E(\wt\vep_{j, t}^2)+5C_{\vep}C^{-1/2}}\right\} + 
11C_\alpha d\left[1+\frac{5\{\E(\wt\vep_{j, t}^2)\}^{2/5}\sqrt{d}}{\log\,T}\right] \times
\\
& \exp\left\{-\frac{4\log\,T}{5c_q\log(1/\mu)}\right\} \le C_0T^{3/2}\log^{-1}T \cdot T^{-C_1}
\end{flalign*}
for some $C_0, C_1 > 0$, where the latter depends on $\mu$, $c_q, C_{\vep}, C$ and $\sigma^{*2}/\sigma_*^2$.
More specifically, we can impose appropriate conditions on the above parameters in order that $C_1 > 7/2+\omega$.
Therefore, $\pr(\cE_1^c)$ is bounded from the above by
$nT^2 \cdot C_0T^{3/2}\log^{-1}T \cdot T^{-C_1} \to 0$ as $T \to \infty$. 
\end{proof}

\begin{lem}
\label{lem:two}
Define the set $\cI_2 = \{(s, b, e):\, 1 \le s < b < e \le T, \  (b-s+1) \wedge (e-b) > d_T\}$
with $d_T$ given in Lemma \ref{lem:one},
and the event 
$\cE_2 = \{\max_{1 \le j \le n}\max_{(s, b, e)\in\cI_2} |\cC_b(\{\wt\vep_{j, t}\}_{t=s}^e)| \le \sqrt{2}\log\,T\}$.
Then as $T \to \infty$, we have $\pr(\cE_2) \to 1$.
\end{lem}
\begin{proof}
Note that $\cI_2 \subset \cI_1$, and
\begin{align*}
|\cC_b(\{\wt\vep_{j, t}\}_{t=s}^e)|
&\le \sqrt{\frac{e-b}{(e-s+1)(b-s+1)}}\left\vert\sum_{t=s}^b \wt\vep_{j, t} \right\vert +
\\
& \sqrt{\frac{b-s+1}{(e-s+1)(e-b)}}\left\vert\sum_{t=b+1}^e \wt\vep_{j, t} \right\vert = I + II.
\end{align*}
On the event $\cE_1$ defined in Lemma \ref{lem:one},
\begin{eqnarray*}
I = \sqrt{\frac{e-b}{e-s+1}} \cdot \frac{1}{\sqrt{b-s+1}}\left\vert \sum_{t=s}^b \wt\vep_{j, t} \right\vert 
\le \sqrt{\frac{e-b}{e-s+1}}\log\,T,
\end{eqnarray*}
and similarly $II \le \sqrt{(b-s+1)/(e-s+1)} \log\,T$, uniformly over all $j$ and $(s, b, e)\in\cI_2$. Hence
\begin{eqnarray*}
\max_{1 \le j \le n}\max_{(s, b, e)\in\cI_2} |\cC_b(\{\wt\vep_{j, t}\}_{t=s}^e)| 
&\le& \max_{(s, b, e)\in\cI_2} \left\{ \sqrt{\frac{e-b}{e-s+1}} + \sqrt{\frac{b-s+1}{e-s+1}} \right\}\log\,T
\\
&\le& \sqrt{2}\log\,T,
\end{eqnarray*}
with probability tending to one as $T \to \infty$, and therefore
$\pr(\cE_2^c) \le \pr(\cE_2^c|\cE_1)\pr(\cE_1) + \pr(\cE_1^c) \to 0$.
\end{proof}

Next, we introduce $N$ additive models $y_{r, t} = g_{r, t} + \xi_{r, t}$, $r=1, \ldots, N$,
which play a vital role in the following proofs.
For each $r$, let $\{k^r_1, \ldots, k^r_n\}$ denote a permutation of $\{1, \ldots, n\}$, 
and $\{i^r_1, \ldots, i^r_n\}$ a set of signs (taking values from \{-1, 1\} with repetitions),
for which the followings hold:
\begin{center}
$i^r_1 \cdot \delta_{k^r_1, r} \ge i^r_2 \cdot \delta_{k^r_2, r} \ge \ldots \ge i^r_n \cdot \delta_{k^r_n, r} \ge 0$
\end{center}
(since $\delta_{k^r_j, r} = f_{k^r_j, \eta_r+1} - f_{k^r_j, \eta_r} = 0$ for all $j \ge m_r+1$, the ordering and the set of signs are not unique).
Then $y_{r, t}$, $g_{r, t}$ and $\xi_{r, t}$ are defined as
\begin{align}
y_{r, t} &= \left\{\frac{m_r(2n-m_r)}{2n}\right\}^\varphi\left\{\frac{1}{m_r}\sum_{j=1}^{m_r}i^r_j \cdot \wt{x}_{k^r_j, t} -
\frac{1}{2n-m_r}\sum_{j=m_r+1}^{n}i^r_j \cdot \wt{x}_{k^r_j, t}\right\},
\label{def:y:r}
\\
g_{r, t} &= \left\{\frac{m_r(2n-m_r)}{2n}\right\}^\varphi\left\{\frac{1}{m_r}\sum_{j=1}^{m_r}i^r_j \cdot \wt{f}_{k^r_j, t} -
\frac{1}{2n-m_r}\sum_{j=m_r+1}^{n}i^r_j \cdot \wt{f}_{k^r_j, t}\right\},
\label{def:g:r}
\\
\xi_{r, t} &= \left\{\frac{m_r(2n-m_r)}{2n}\right\}^\varphi\left\{\frac{1}{m_r}\sum_{j=1}^{m_r}i^r_j \cdot \wt\vep_{k^r_j, t} -
\frac{1}{2n-m_r}\sum_{j=m_r+1}^{n}i^r_j \cdot \wt\vep_{k^r_j, t}\right\}.
\label{def:xi:r}
\end{align}
By its definition, $\{g_{r, t}\}_{t=1}^T$ is a piecewise constant signal with change-points at $t=\eta_1, \ldots, \eta_N$, 
and its jump at $t=\eta_r$ is of size satisfying the following:
\beq
\label{min:jump}
g_{r, \eta_r+1} - g_{r, \eta_r} = \left\{\frac{m_r(2n-m_r)}{2n}\right\}^\varphi\frac{1}{m_r}\sum_{j\in\Pi_r}\frac{|\delta_{j, r}|}{\wt\sigma_j} \ge 
\left\{\frac{m_rn}{2n}\right\}^\varphi\frac{2\wt{\delta}_r}{3\sigma^*} \ge c_1m_r^\varphi \wt{\delta}_r. 
\eeq

Then Lemma \ref{lem:one} implies that for all $r = 1, \ldots, N$,
\begin{eqnarray}
\label{xi:bound:one}
\max_{(s, e) \in \cI_1}
\frac{1}{\sqrt{e-s+1}} \left\vert \sum_{t=s}^e\xi_{r, t} \right\vert \le 2m_r^\varphi\log\,T
\le 2n^\varphi\log\,T
\end{eqnarray}
with probability converging to one, since within the event $\cE_1$ of Lemma \ref{lem:one},
\begin{align*}
\frac{1}{\sqrt{e-s+1}}\left\vert\sum_{t=s}^e\xi_{r, t}\right\vert
&\le \left\{\frac{m_r(2n-m_r)}{2n}\right\}^\varphi\left\{\frac{1}{m_r}\sum_{j=1}^{m_r}\left\vert \frac{1}{\sqrt{e-s+1}}\sum_{t=s}^e \vep_{k^r_j, t} \right\vert \right.
\\
& + \left.\frac{1}{2n-m_r}\sum_{j=m_r+1}^{n}\left\vert \frac{1}{\sqrt{e-s+1}}\sum_{t=s}^e \vep_{k^r_j, t} \right\vert\right\} 
\le 2m_r^\varphi\log\,T.
\end{align*}
Similarly, Lemma \ref{lem:two} implies that for all $r$,
\begin{eqnarray}
\label{xi:bound:two}
\max_{(s, b, e)\in\cI_2} |\cC_b(\{\xi_{r, t}\}_{t=s}^e)| \le 2\sqrt{2}m_r^\varphi\log\,T \le 2\sqrt{2}n^\varphi\log\,T
\end{eqnarray}
with probability tending to one. 

Now we consider a generic additive model
\begin{eqnarray}
\label{def:add}
y_t = g_t + \xi_t,
\end{eqnarray}
where $y_t$, $g_t$ and $\xi_t$ are obtained from the panel data $\{x_{j, t}\}_{t=1}^T, \ j=1, \ldots, n$
in the same manner as $y_{r, t}$, $g_{r, t}$ and $\xi_{r, t}$ in (\ref{def:y:r})--(\ref{def:xi:r}),
with respect to some $m\in\{1, \ldots, n\}$, a permutation of index $\{k_1, \ldots, k_n\}$ and a sign sequence $\{i_1, \ldots, i_n\}$.
Then $g_t$ is a piecewise constant signal with change-points at
$t=\eta_1, \ldots, \eta_N$ satisfying (B1).
Also, it is easily seen that the inequalities (\ref{xi:bound:one})--(\ref{xi:bound:two})
hold with the zero-mean noise series $\xi_t$ in place of $\xi_{r, t}$
with probability converging to one.

Recall that $s$ and $e$ denote the start and the end of an interval
which is examined at some stage of our search for the change-points.
Let $s$ and $e$ satisfy
\begin{eqnarray}
\label{lem:cond:zero}
\eta_{q_1} \le s < \eta_{q_1+1} < \ldots < \eta_{q_2} < e \le \eta_{q_2+1}
\end{eqnarray}
for $0 \le q_1 < q_2 \le N$.
In some of the following lemmas, we impose at least one of following conditions:
\begin{eqnarray}
s< \eta_q-c_2T^\beta < \eta_q+c_2T^\beta < e \mbox{ for some } q\in\{q_1+1, \ldots, q_2\},
\label{lem:cond:one} \\
\{(\eta_{q_1+1}-s)\wedge(s-\eta_{q_1})\} \vee \{(\eta_{q_2+1}-e)\wedge(e-\eta_{q_2})\} \le c_3\bar{\epsilon}_T,
\label{lem:cond:two}
\end{eqnarray}
with $\bar{\epsilon}_T$ to be defined later.
The condition (\ref{lem:cond:one}) implies that there is at least one change-point to be detected
which is sufficiently distanced away from the previously detected change-points $s$ and $e$,
and (\ref{lem:cond:two}) indicates that each of $s$ and $e$ is detected for one of the true change-points.

Since the CUSUM statistics are not affected by the shift in the overall level of $g_t$,
we assume that $\sum_{t=s}^eg_t = 0$ without loss of generality.
Then the CUSUM statistics computed on $\{g_t\}_{t=s}^e$ can be re-written as
$\cC_b(\{g_t\}_{t=s}^e) = \sqrt{\frac{e-s+1}{(b-s+1)(e-b)}}\sum_{t=s}^bg_t$, $b=s, \ldots, e-1$.

\begin{lem}
\label{lem:three}
For $s$ and $e$ satisfying (\ref{lem:cond:zero}),
there exists $q'\in\{q_1+1, \ldots, q_2\}$ which satisfies
$\eta_{q'} = \arg\max_{b\in[s, e)}|\cC_b(\{g_t\}_{t=s}^e)|$.
\end{lem}
\begin{proof} 
The proof follows directly from Lemmas 2.2--2.3 of \cite{venkatraman1992}.
\end{proof}

\begin{lem}
\label{lem:four}
Let (\ref{lem:cond:one}) and (\ref{lem:cond:two}) hold. 
For $\heta = \arg\max_{b\in[s, e)}|\cC_b(\{y_t\}_{t=s}^e)|$, 
there exists a true change-point $\eta_q \equiv \eta\in(s, e)$ satisfying $|\heta - \eta| < c_0\bar{\epsilon}_T$ 
with probability converging to one, provided that
\beqs
&& |\cC_{\eta}(\{g_t\}_{t=s}^e)||\cC_{\eta}(\{g_t\}_{t=s}^e)-\cC_{\heta}(\{g_t\}_{t=s}^e)|
> C_1n^\varphi\log\,T \times \\
&& \max\left\{\begin{array}{l}
n^\varphi\log\,T, 
\\
\textstyle{\sqrt{\eta-s+1}\left\vert\frac{1}{\heta-s+1}\sum_{t=s}^{\heta}g_t - \frac{1}{\eta-s+1}\sum_{t=s}^{\eta}g_t\right\vert},
\\
\textstyle{\sqrt{\bar{\epsilon}_T}\left\vert\frac{1}{\heta-s+1}\sum_{t=s}^{\heta}g_t - \frac{1}{e-\eta}\sum_{t=\eta+1}^e g_t \right\vert}
\end{array}\right\} \quad \mbox{for some } C_1>0.
\eeqs
\end{lem}
\begin{proof}
The following proof is an adaptation of the proof of Theorem \ref{thm:two}.1 in \cite{piotr2014} to a non-i.i.d. case.

On the segment $[s, e]$, detecting a change-point is equivalent to fitting the best step function $\wh{g}_t$
(a piecewise constant function with one change-point)
which minimises $\sum_{t=s}^e(y_t-h_t)^2$ among all step functions $h_t$ defined on $[s, e]$.
Let $g_t^*$ be the best step function approximation to $g_t$ on $[s, e]$,
which may not be unique.
From Lemma \ref{lem:three}, $g_t^*$ needs to have its change-point $\eta$ 
coincide with one of the true change-points $\eta_q, \, q\in\{q_1+1, \ldots, q_2\}$.

Let us assume that $|\heta-\eta| = c_0\bar{\epsilon}_T$.
Due to the fact that $|\cC_b(\{g_t\}_{t=s}^e)|$ is monotonic, or decreasing and then increasing in $b$
between two adjacent change-points of $g_t$ (Lemma 2.7 of \cite{venkatraman1992}),
it is enough to consider the case when $\heta$ satisfies $|\heta - \eta| = c_0\bar{\epsilon}_T$.
Then, if it is shown that
\begin{eqnarray}
\label{lem:four:eq:one}
\sum_{t=s}^e(y_t-g_t^*)^2 - \sum_{t=s}^e(y_t-\wh{g}_t)^2 < 0,
\end{eqnarray}
the claim would be proved to be contradiction.
Expanding the LHS of (\ref{lem:four:eq:one}), we obtain
\begin{align*}
& \sum_{t=s}^e (\xi_t+g_t-g_t^*)^2 - \sum_{t=s}^e (\xi_t+g_t-\wh{g}_t)^2
= 2\sum_{t=s}^e \xi_t(\wh{g}_t-g_t^*) 
\\
& + \sum_{t=s}^e \{(g_t-g_t^*)^2 - (g_t-\wh{g}_t)^2\} 
= I + II.
\end{align*}
Clearly, $II<0$ from the definition of $g_t^*$.

Let $\boldsymbol{\Psi}$ be the set of vectors of length $(e-s+1)$ whose elements are initially positive and constant,
then after a break, are negative and constant;
moreover, the elements sum to zero and when squared, sum to one.
Since we assume that $\sum_{t=s}^e g_t = 0$,
we can find a vector $\bpsi^*\in\boldsymbol{\Psi}$ satisfying $\bg^* = \inner{\bg}{\bpsi^*}\bpsi^*$
where $\bg=(g_s, \ldots, g_e)^\top$ and $\bg^*=(g_s^*, \ldots, g_e^*)^\top$.
Then we have
\begin{eqnarray*}
\sum_{t=s}^e(g_t-g_t^*)^2
= \sum_{t=s}^e g_t^2 - \inner{\bg}{\bpsi^*}^2.
\end{eqnarray*}

Let a step function $\wt{g}_t$ be chosen to minimise $\sum_{t=s}^e(g_t - h_t)^2$
among all the step functions $h_t$ defined on $[s, e]$, under the constraint that $h_t$ has its change-point at $t=\heta$.
For such $\wt{g}_t$, we have $\sum_{t=s}^e(g_t-\wt{g}_t)^2 \le \sum_{t=s}^e(g_t-\wh{g}_t)^2$.
Again, there exists a vector $\wt{\bpsi}\in\boldsymbol{\Psi}$ satisfying $\wt{\bg} = \inner{\bg}{\wt{\bpsi}}\wt{\bpsi}$
with $\wt{\bg} = (\wt{g}_s, \ldots, \wt{g}_e)^\top$.
Then
\begin{eqnarray}
|II| &\ge& \sum_{t=s}^e\{(g_t-\wt{g}_t)^2 - (g_t-g_t^*)^2\} = \inner{\bg}{\bpsi^*}^2 - \inner{\bg}{\wt{\bpsi}}^2 
\nonumber
\\
&=& (\inner{\bg}{\bpsi^*}+\inner{\bg}{\wt{\bpsi}})(\inner{\bg}{\bpsi^*}-\inner{\bg}{\wt{\bpsi}})
\ge |\inner{\bg}{\bpsi^*}||\inner{\bg}{\bpsi^*} - \inner{\bg}{\wt{\bpsi}}|,
\label{lem:four:eq:two}
\end{eqnarray}
since $|\cC_\eta(\{g_t\}_{t=s}^e)| = |\inner{\bg}{\bpsi^*}| \ge |\inner{\bg}{\wt{\bpsi}}| = |\cC_{\heta}(\{g_t\}_{t=s}^e)|$.

Turning to $I$, it is decomposed as
\begin{eqnarray*}
2\sum_{t=s}^e \xi_t(\wh{g}_t-g_t^*) = 2\sum_{t=s}^e \xi_t(\wh{g}_t-\wt{g}_t) + 2\sum_{t=s}^e \xi_t(\wt{g}_t-g_t^*),
\end{eqnarray*}
and each of the two terms are split into sub-sums computed over the intervals where $\wh{g}_t-\wt{g}_t$ and $\wt{g}_t-g_t^*$ are constant, respectively.
Letting $\heta > \eta$ without loss of generality, we have
\begin{eqnarray*}
\sum_{t=s}^e \xi_t(\wt{g}_t-g^*_t) = \left(\sum_{t=s}^\eta + \sum_{t=\eta+1}^{\heta} + \sum_{t=\heta+1}^e\right)\xi_t(\wt{g}_t-g^*_t) = III + IV + V.
\end{eqnarray*}
Then
$|III| \le 2n^\varphi\log\,T\sqrt{\eta-s+1}
\vert (\heta-s+1)^{-1}\sum_{t=s}^{\heta}g_t - (\eta-s+1)^{-1}\sum_{t=s}^{\eta}g_t\vert$
with probability tending to one, following (\ref{xi:bound:one}).
$|V|$ is of the same order as $|III|$, and similar arguments lead to
$|IV| \le 2n^\varphi\log\,T \sqrt{\heta-\eta+1}
\vert(\heta-s+1)^{-1}\sum_{t=s}^{\heta}g_t - (e-\eta)^{-1}\sum_{t=\eta+1}^e g_t\vert$.

As for $\sum_{t=s}^e\xi_t(\wh{g}_t-\wt{g}_t)$, we have
\begin{eqnarray*}
\sum_{t=s}^e\xi_t(\wh{g}_t-\wt{g}_t) =
\left(\sum_{t=s}^{\heta}+\sum_{t=\heta+1}^{e}\right)\xi_t(\wh{g}_t-\wt{g}_t)
= VI+VII.
\end{eqnarray*}
$|VI|$ and $|VII|$ are of the same order, and with probability converging to one,
\begin{eqnarray*}
|VI| = |\sum_{t=s}^{\heta}\xi_t| \cdot \frac{1}{\heta-s+1}\left\vert\sum_{t=s}^{\heta}(y_t-g_t)\right\vert
= \frac{1}{\heta-s+1}\left(\sum_{t=s}^{\heta}\xi_t\right)^2 \le 4n^{2\varphi}\log^2\,T.
\end{eqnarray*}
Putting together (\ref{lem:four:eq:two}) and the upper bound on $|III|$--$|VII|$,
we have the dominance of $II$ over $I$ under the conditions given in the lemma.
\end{proof}

\subsection{Proofs of Theorems \ref{thm:zero}--\ref{thm:one}}
\label{sec:pf:thm:one}

Throughout the section, we assume (A1)--(A3) and (A6), and (A4)--(A5) where applicable.
In the problem of detecting (at most) a single change-point, 
(\ref{lem:cond:one})--(\ref{lem:cond:two}) are met with $s=1$ and $e=T$ under (A4) with 
$\bar{\epsilon}_T = (n^{\varphi}/m_1)^2\wt{\delta}_1^2\log^2\,T$. 
For $(\heta_1, \wh{m}_1) = \arg\max_{b\in[1, T]\setminus\cI_{1, T}, \, 1 \le m \le n} \cD^\varphi_m(\{|\cX^{(j)}_{1, b, T}|\}_{j=1}^n)$,
let $\{k^0_1, \ldots, k^0_n\}$ denote a permutation of $\{1, 2, \ldots, n\}$
satisfying
$|\cX^{k^0_1}_{1, \heta_1, T}| \ge |\cX^{k^0_2}_{1, \heta_1, T}| \ge \ldots \ge |\cX^{k^0_n}_{1, \heta_1, T}|$,
and $i^0_j\in\{-1, 1\}$ satisfy
$|\cX^{k^0_j}_{1, \heta_1, T}| = i^0_j \cdot \cX^{k^0_j}_{1, \heta_1, T}$
for all $j=1, \ldots, n$.
With $\wh{m}_1$, $\{k^0_j\}$ and $\{i^0_j\}$ replacing $m_r$, $\{k^r_j\}$ and $\{i^r_j\}$, respectively,
we have an additive model $y_{0, t} = g_{0, t} + \xi_{0, t}$
with its components defined in the same manner as
$y_{r, t}$, $g_{r, t}$ and $\xi_{r, t}$ in (\ref{def:y:r})--(\ref{def:xi:r}).
Then $g_{0, t}$ is piecewise constant with (at most) one change-point at $t=\eta_1$,
and $\{\xi_{0, t}\}$ satisfies (\ref{xi:bound:one})--(\ref{xi:bound:two}) in place of $\{\xi_{r, t}\}$.
Note that the DC statistic at $m=\wh{m}_1$ and $b=\heta_1$ 
can equivalently be represented using $\{y_{0, t}\}$ as
\begin{eqnarray*}
\cD^\varphi_{\wh{m}_1}(\{|\cX^{(j)}_{1, \heta_1, T}|\}_{j=1}^n) =
\cD^\varphi_{\wh{m}_1}(\{i^0_j \cdot \cX^{k^0_j}_{1, \heta_1, T}\}_{j=1}^n) = \cC_{\heta_1}(\{y_{0, t}\}_{t=1}^T).
\end{eqnarray*}

In the case when there exists no change-point, $g_{0, t}$ is constant and therefore
$\cC_b(\{g_{0, t}\}_{t=1}^T) = 0$ for all $b$.
As $(1, \heta_1, T)\in\cI_2$, we have $\cD^\varphi_{\wh{m}_1}(\{|\cX^{(j)}_{1, \heta_1, T}|\}_{j=1}^n) = \cC_{\heta_1}(\{\xi_{0, t}\}_{t=1}^T) < C'n^\varphi\log\,T$
for some $C'>2\sqrt{2}$, which proves Theorem \ref{thm:zero}.

To prove Theorem \ref{thm:one}, we need additional lemmas stated with the generic additive model in (\ref{def:add}).
\vspace{5pt}
\begin{lem}
\label{lem:one:one}
Assume that there exists a single change-point $\eta_1$ in $g_t$ which satisfies (A4)
and $|g_{\eta_1+1}-g_{\eta_1}| \ge \delta$.
Then for some $C_2>0$,
\begin{eqnarray}
\label{lem:one:one:eq}
|\cC_{\eta_1}(\{g_t\}_{t=1}^T)| = \max_{b\in[1, T)}|\cC_b(\{g_t\}_{t=1}^T)| = 
\sqrt{\frac{\eta_1(T-\eta_1)}{T}}\delta \ge C_2\delta T^{\beta}.
\end{eqnarray}
\end{lem}
\begin{proof}
The first equality (\ref{lem:one:one:eq}) is a direct result of Lemma \ref{lem:three}.
The second equality follows from the definition of $|\cC_{\eta_1}(\{g_t\}_{t=1}^T)|$.
\end{proof}

\begin{lem}
\label{lem:one:two}
Assume that the conditions imposed in Lemma \ref{lem:one:one} are met. 
Then for some $\ept$, we have
\begin{eqnarray*}
|\cC_{\eta_1}(\{g_t\}_{t=1}^T) - \cC_{b'}(\{g_t\}_{t=1}^T)| \ge C_3\ept|\cC_{\eta_1}(\{g_t\}_{t=1}^T)|\frac{T}{\eta_1(T-\eta_1)}
\end{eqnarray*}
with any $b'$ satisfying $|b'-\eta_1| \ge c_0\ept$.
\end{lem}
\begin{proof}
Without loss of generality, let $g_{\eta_1} = g^*_1 > g^*_2 = g_{\eta_1+1}$ and that $b' > \eta_1$. 
Then
\begin{align*}
& \textstyle{\cC_{\eta_1}(\{g_t\}_{t=1}^T) - \cC_{b'}(\{g_t\}_{t=1}^T) \ge  \sqrt{\frac{\eta_1(T-\eta_1)}{T}}\delta
- \sqrt{\frac{(\eta_1+c_0\ept)(T-\eta_1-c_0\ept)}{T}}\left(\frac{g^*_1\eta_1+g^*_2c_0\ept}{\eta_1+c_0\ept}-g^*_2\right)}
\\
&= \textstyle{\sqrt{\frac{\eta_1(T-\eta_1)}{T}}\delta\left(1-\sqrt{\frac{\eta_1(T-\eta_1-c_0\ept)}{(\eta_1+c_0\ept)(T-\eta_1)}}\right)
= \cC_{\eta_1}(\{g_t\}_{t=1}^T) \cdot \frac{\sqrt{1+\frac{c_0\ept}{\eta_1}}-\sqrt{1-\frac{c_0\ept}{T-\eta_1}}}{\sqrt{1+\frac{c_0\ept}{\eta_1}}}}
\\
&\ge \textstyle{\cC_{\eta_1}(\{g_t\}_{t=1}^T) \cdot \frac{c_0\ept}{2\sqrt{2}}\frac{T}{\eta_1(T-\eta_1)},}
\end{align*}
where the inequality follows from the Taylor expansion.
\end{proof}

Lemma \ref{lem:three} implies that
$|\cC_{\eta_1}(\{g_{0, t}\}_{t=1}^T)| \ge |\cC_{\heta_1}(\{g_{0, t}\}_{t=1}^T)|$
while $|\cC_{\eta_1}(\{y_{0, t}\}_{t=1}^T)| \le 
|\cC_{\heta_1}(\{y_{0, t}\}_{t=1}^T)|$ by assumption.
From (\ref{xi:bound:two}), we have
$|\cC_b(\{g_{r, t}\}_{t=1}^T)| \ge |\cC_b(\{y_{r, t}\}_{t=1}^T)| - c_4n^\varphi\log\,T$ 
at any $b\in\cI_{1, T}$ and $r=0, 1$ for some $c_4 > 0$.
Note that for all $\varphi\in[0, 1]$ and $1 \le m \le n$, on a given interval $[s, e]$, 
the ordering and the set of signs 
that are applied to $\{\cX^j_{1, b, T}\}_{j=1}^n$ in order to produce $\{|\cX^{(j)}_{1, b, T}|\}_{j=1}^n$
(identical to $\{k^0_j\}$ and $\{i^0_j\}$ when $b=\heta_1$),
leads to the maximum value of $\cD^\varphi_m$ at any $b$,
among all possible index permutations and the sets of signs.
Based on the above observations, the following holds with probability tending to one:
\beq
&& |\cC_{\heta_1}(\{y_{0, t}\}_{t=1}^T)| \ge
\max_{1 \le m \le n} \cD^\varphi_m(\{|\cX^{(j)}_{1, \eta_1, T}|\}_{j=1}^n)
\ge
|\cC_{\eta_1}(\{y_{1, t}\}_{t=1}^T)|, \mbox{ and}
\nonumber \\
&& 
|\cC_{\wh{\eta}_1}(\{g_{0, t}\}_{t=1}^T)| \ge
|\cC_{\eta_1}(\{g_{1, t}\}_{t=1}^T)| - 2c_4n^\varphi \log\,T
\ge
\frac{c_1m_1^\varphi\wt{\delta}_1}{2}\sqrt{\frac{\eta_1(T-\eta_1)}{T}}, \label{ineq:delta}
\eeq
where the last inequality of (\ref{ineq:delta}) follows from (\ref{min:jump}) and Lemma \ref{lem:one:one}. 
Recalling the rates given above Theorem \ref{thm:zero} for $\pi^\varphi_{n, T}$, 
we have $\cD^\varphi_{\wh{m}_1}(\{|\cX^{(j)}_{1, \heta_1, T}|\}_{j=1}^n) > \pi^\varphi_{n, T}$.

To prove the consistency in the location of the estimated change-point, let $|\heta_1-\eta_1|=c_0\ept$
and assume that $\heta_1 > \eta_1$ without loss of generality.
We can view the problem of deriving the upper bound on the bias $|\heta_1-\eta_1|$ for 
$\heta_1 = \arg\max_{b\in[1, T]\setminus\cI_{1, T}}\max_{1 \le m \le n} \cD^\varphi_m(\{|\cX^{(j)}_{1, b, T}|\}_{j=1}^n)$, 
as that for 
$$\heta_1 = \arg\max_{b\in[1, T]\setminus\cI_{1, T}} \cC_b(\{y_{0, t}\}_{t=1}^T).$$ 
Adopting the notations in Lemma \ref{lem:one:two} with $g_{0, t}$ replacing $g_t$,
\beqs
\left\vert\frac{1}{\heta_1}\sum_{t=1}^{\heta_1}g_{0, t} - \frac{1}{\eta_1}\sum_{t=s}^{\eta_1}g_{0, t}\right\vert
&=& \left\vert\frac{g^*_1\eta_1+g^*_2c_0\ept}{\heta_1} - g^*_1\right\vert \le 
|g^*_1-g^*_2|\frac{c_0\ept}{\eta_1},
\\
\left\vert\frac{1}{\heta_1}\sum_{t=1}^{\heta_1}g_{0, t} - \frac{1}{T-\eta_1}\sum_{t=\eta_1+1}^e g_{0, t}\right\vert
&=& \left\vert\frac{g^*_1\eta_1+g^*_2c_0\ept}{\heta_1} - g^*_2\right\vert \le |g^*_1-g^*_2|,
\eeqs
and (\ref{ineq:delta}) indicates that $|g^*_1-g^*_2| \ge c_5m_1^\varphi\wt{\delta}_1$ for some $c_5 > 0$.
Combining the above bounds with Lemmas \ref{lem:one:one}--\ref{lem:one:two},
the conditions in Lemma \ref{lem:four} are met for $\eta=\eta_1$, $\heta=\heta_1$, $s=1$, $e=T$ and 
$\ept = (n^{\varphi}/m_1)^2\wt{\delta}_1^2\log^2\,T$,
and thus Theorem \ref{thm:one} is proved.  \hfill $\square$

\subsection{Proof of Theorem \ref{thm:two}}
\label{sec:pf:thm:two}

Throughout the section, we assume (A1)--(A3), (A6) and (B1)--(B2).
Let $s, e$ satisfy (\ref{lem:cond:one})--(\ref{lem:cond:two}) with $\bar{\epsilon}_T = n^{2\varphi}\uDel_\varphi^{-2} T^{5(1-\beta)}\log^2\,T$. 
For 
$$ (\heta, \wh{m}) = \arg\max_{b\in[s, e]\setminus\cI_{s, e}, \, 1 \le m \le n} \cD^\varphi_m(\{|\cX^{(j)}_{s, b, e}|\}_{j=1}^n), $$
let $\{k^0_1, \ldots, k^0_n\}$ denote a permutation of $\{1, 2, \ldots, n\}$ satisfying
$|\cX^{k^0_1}_{s, \heta, e}| \ge |\cX^{k^0_2}_{s, \heta, e}| \ge \ldots \ge |\cX^{k^0_n}_{s, \heta, e}|$,
and $i^0_j\in\{-1, 1\}$ satisfy
$|\cX^{k^0_j}_{s, \heta, e}| = i^0_j \cdot \cX^{k^0_j}_{s, \heta, e}$
for all $j=1, \ldots, n$.
As before, with $\wh{m}$, $\{k^0_j\}$ and $\{i^0_j\}$ replacing $m_r$, $\{k^r_j\}$ and $\{i^r_j\}$, respectively,
we define an additive model $y_{0, t} = g_{0, t} + \xi_{0, t}$
with its components obtained in the same manner as those in (\ref{def:y:r})--(\ref{def:xi:r}).
Then $g_{0, t}$ is a piecewise constant signal with change-points at $t=\eta_r, \ r=1, \ldots, N$,
and $\{\xi_{0, t}\}$ satisfies (\ref{xi:bound:one})--(\ref{xi:bound:two}) in place of $\{\xi_{r, t}\}$.
Also, the DC statistic at $m = \wh{m}$ and $b = \heta$ 
can equivalently be represented by
$\cD^\varphi_{\wh{m}}(\{|\cX^{(j)}_{s, \heta, e}|\}_{j=1}^n) =
\cD^\varphi_{\wh{m}}(\{i^0_j \cdot \cX^{k^0_j}_{s, \heta, e}\}_{j=1}^n) = \cC_{\heta}(\{y_{0, t}\}_{t=s}^e)$.

Below we introduce additional lemmas stated with the generic additive model in (\ref{def:add}).
\begin{lem}
\label{lem:two:one}
Let $s$ and $e$ satisfy (\ref{lem:cond:one})
and assume that there exists a change-point $\eta_q\in(s, e)$ at which
$(\eta_q-s+1) \wedge (e-\eta_q) > c_2T^\beta$
and $|g_{\eta_q+1}-g_{\eta_q}| \ge \delta$.
Then there exists $q'\in\{q_1+1, \ldots, q_2\}$ and $C_4>0$ such that
$|\cC_{\eta_{q'}}(\{g_t\}_{t=s}^e)| = \max_{b\in[s, e)}|\cC_b(\{g_t\}_{t=s}^e)| \ge C_4\delta T^{\beta-1/2}$.
\end{lem}
\begin{proof}
The equality part is a direct result of Lemma \ref{lem:three}.
Since $|g_{\eta_q+1}| \vee |g_{\eta_q}| \ge \delta/2$,
we have 
$\vert\sum_{t=\eta_q-c_2T^\beta+1}^{\eta_q} g_t \vert \vee \vert\sum_{t=\eta_q+1}^{\eta_q+c_2T^\beta} g_t \vert
\ge c_2\delta T^\beta/2$.
Hence, \\
$\max_{b\in[s, e)}|\sum_{t=s}^b g_t| \ge c_2\delta T^\beta/4$,
from which it is derived that
\begin{eqnarray*}
\max_{b\in[s, e)}|\cC_b(\{g_t\}_{t=s}^e)| \ge
\min_{b\in[s, e)}\sqrt{\frac{e-s+1}{(b-s+1)(e-b)}} \cdot
\max_{b\in[s, e)}|\sum_{t=s}^b g_t|
\ge C_4\delta T^{\beta-1/2}.
\end{eqnarray*}
\end{proof}

\begin{lem}
\label{lem:two:two}
Assume that $s$ and $e$ meet the conditions (\ref{lem:cond:one})--(\ref{lem:cond:two})
and let $\eta\in(s, e)$ denote a change-point that satisfies
\begin{eqnarray}
\label{lem:two:two:cond}
|\cC_\eta(\{g_t\}_{t=s}^e)| > \max_{b\in[s, e)} |\cC_b(\{g_t\}_{t=s}^e)| - C_5n^\varphi\log\,T
\end{eqnarray}
for some positive constant $C_5$.
Then for some $\ept$ and $C_6 > 0$, we have
\begin{eqnarray}
\label{lem:two:two:eq}
|\cC_\eta(\{g_t\}_{t=s}^e) - \cC_{b'}(\{g_t\}_{t=s}^e)| > C_6T^{\beta-2}\ept |\cC_\eta(\{g_t\}_{t=s}^e)|
\end{eqnarray}
with any $b'$ satisfying $|b'-\eta| \ge \ept$.
\end{lem}
\begin{proof}
The result is a modification of Lemma 2.6 in \cite{venkatraman1992} and the arguments therein are directly applicable
to show (\ref{lem:two:two:eq}). 
\end{proof}

For any interval $[s, e]$, we define an index set $\cR_{s, e} \subset \{1, \ldots, N\}$
as $\cR_{s, e} = \{1 \le r \le N: \, \eta_r\in[s, e]\setminus\cI_{s, e}\}$.
Adopting the same arguments as in Section \ref{sec:pf:thm:one}, 
\begin{align}
& |\cC_{\heta}(\{y_{0, t}\}_{t=s}^e)| \ge
\max_{q\in\cR_{s, e}}\max_{1 \le m \le n} \cD^\varphi_m(\{|\cX^{(j)}_{s, \eta_q, e}|\}_{j=1}^n)
\ge
\max_{q\in\cR_{s, e}} \max_{1 \le r \le N}|\cC_{\eta_q}(\{y_{r, t}\}_{t=s}^e)|, 
\nonumber
\\
& |\cC_{\heta}(\{g_{0, t}\}_{t=s}^e)| \ge
\max_{q\in\cR_{s, e}}\max_{1 \le r \le N} |\cC_{\eta_q}(\{g_{r, t}\}_{t=s}^e)| - 2c_4n^\varphi \log\,T
\ge
\frac{C_4}{2} c_1\uDel_\varphi T^{\beta-1/2},
\label{ineq:delta:bs}
\end{align}
where the last inequality of (\ref{ineq:delta:bs}) follows from (\ref{min:jump}), (\ref{lem:cond:one}) and Lemma \ref{lem:two:one}.
\vspace{5pt}
\begin{lem}
\label{lem:two:three}
Let (\ref{lem:cond:one}) and (\ref{lem:cond:two}) hold. 
For $\heta = \arg\max_{b\in[s, e]\setminus\cI_{s, e}}|\cC_b(\{y_{0, t}\}_{t=s}^e)|$, 
there exists a true change-point $\eta_q \equiv \eta\in(s, e)$ satisfying $|\heta - \eta| < c_0\ept$ 
with probability converging to one, where
$\ept = n^{2\varphi}\uDel_\varphi^{-2}T^{5(1-\beta)}\log^2\,T$.
\end{lem}
\begin{proof}
We adopt the notations from the proof of Lemma \ref{lem:four} with $g_{0, t}$ in place of $g_t$.

Recall that from Lemma \ref{lem:three}, $g_t^*$ needs to have its change-point $\eta$ 
coincide with one of the true change-points $\eta_{q_1+1}, \ldots, \eta_{q_2}$.
Trivially, such $\eta$ satisfies (\ref{lem:two:two:cond}) since
$|\cC_\eta(\{g_{0, t}\}_{t=s}^e)| = \max_{b\in[s, e)} |\cC_b(\{g_{0, t}\}_{t=s}^e)|$.
Under (\ref{lem:cond:one})--(\ref{lem:cond:two}),
we have either $(\eta-s+1) \wedge (e-\eta) < c_3\ept$ or
$(\eta-s+1) \wedge (e-\eta) > c_2T^\beta$.
If the former is the case, since $|g_{0, t}| \le 2\bar{f}n^{\varphi}$ uniformly in $t$ under (A3),
$|\cC_\eta(\{g_{0, t}\}_{t=s}^e)| \le 4\bar{f}n^{\varphi}(c_3\ept)^{1/2}  < C_4\uDel_\varphi T^{\beta-1/2}$, 
which leads to contradict (\ref{ineq:delta:bs})
and thus $(\eta-s+1) \wedge (e-\eta) > c_2T^\beta$.

Now we turn our attention to bound the terms $|III|$ and $|IV|$ in the presence of multiple change-points. Firstly,
$\textstyle{\vert(\heta-s+1)^{-1}\sum_{t=s}^{\heta}g_{0, t} - (\eta-s+1)^{-1}\sum_{t=s}^{\eta}g_{0, t}\vert}$
\begin{align*}
&= \textstyle{\left\vert \frac{1}{\heta-s+1}\sqrt{\frac{(\heta-s+1)(e-\heta)}{e-s+1}}\cC_{\heta}(\{g_{0, t}\}_{t=s}^e)
- \frac{1}{\eta-s+1}\sqrt{\frac{(\eta-s+1)(e-\eta)}{e-s+1}}\cC_\eta(\{g_{0, t}\}_{t=s}^e)\right\vert}
\\
&= \textstyle{\frac{1}{\sqrt{e-s+1}}\left\vert \sqrt{\frac{e-\heta}{\heta-s+1}}\{\cC_{\heta}(\{g_{0, t}\}_{t=s}^e)-\cC_{\eta}(\{g_{0, t}\}_{t=s}^e)\}\right.}
\\
& \qquad \qquad \textstyle{\left.- \left(\sqrt{\frac{e-\eta}{\eta-s+1}}-\sqrt{\frac{e-\heta}{\heta-s+1}}\right)\cC_\eta(\{g_{0, t}\}_{t=s}^e)\right\vert}
\\
&\le \textstyle{\sqrt{\frac{e-\eta}{(e-s+1)(\eta-s+1)}}
\left\{ \left\vert 1 - \frac{\sqrt{1-\frac{\heta-\eta}{e-\eta}}}{\sqrt{1+\frac{\heta-\eta}{\eta-s+1}}}\right\vert|\cC_{\eta}(\{g_{0, t}\}_{t=s}^e)|+ \right.}
\\
& \qquad \qquad \qquad \qquad \textstyle{\vert\cC_{\heta}(\{g_{0, t}\}_{t=s}^e)-\cC_{\eta}(\{g_{0, t}\}_{t=s}^e)\bigg\vert \bigg\}}
\\
&\le
\textstyle{\sqrt{\frac{e-\eta}{(e-s+1)(\eta-s+1)}}\left\{C_7T^{-\beta}\ept|\cC_{\eta}(\{g_{0, t}\}_{t=s}^e)|+
\left\vert\cC_{\wh\eta}(\{g_{0, t}\}_{t=s}^e)-\cC_{\eta}(\{g_{0, t}\}_{t=s}^e)\right\vert \right\}}
\end{align*}
for some fixed $C_7>0$, where the last inequality follows from the Taylor expansion. 
Since $T^{\beta-2} \le T^{-\beta}$, Lemma \ref{lem:two:two} leads to
$|III| \le 2n^\varphi\ept T^{-\beta}\log\,T|\cC_{\eta}(\{g_{0, t}\}_{t=s}^e)|$ (recall the notation from Lemma \ref{lem:four})
with probability tending to one. Similarly, 
$\textstyle{\vert(\heta-s+1)^{-1}\sum_{t=s}^{\heta}g_{0, t} - (e-\eta)^{-1}\sum_{t=\eta+1}^e g_{0, t}\vert}$
\begin{eqnarray*}
&=& \textstyle{\left\vert \frac{1}{\heta-s+1}\sqrt{\frac{(\heta-s+1)(e-\heta)}{e-s+1}}\cC_{\heta}(\{g_{0, t}\}_{t=s}^e)
+ \frac{1}{e-\eta}\sqrt{\frac{(\eta-s+1)(e-\eta)}{e-s+1}}\cC_\eta(\{g_{0, t}\}_{t=s}^e)\right\vert}
\\
&=& \textstyle{\frac{1}{\sqrt{e-s+1}}\left\vert \left(\sqrt{\frac{e-\heta}{\heta-s+1}}+\sqrt{\frac{\eta-s+1}{e-\eta}}\right)\cC_\eta(\{g_{0, t}\}_{t=s}^e)\right.}
\\
&& 
\qquad \qquad  \textstyle{ + \left.\sqrt{\frac{e-\heta}{\heta-s+1}}\{\cC_{\heta}(\{g_{0, t}\}_{t=s}^e)-\cC_{\eta}(\{g_{0, t}\}_{t=s}^e)\}\right\vert}
\\
&\le&
\textstyle{2\sqrt{\frac{e-s+1}{(\eta-s+1)(e-\eta)}}|\cC_\eta(\{g_{0, t}\}_{t=s}^e)|},
\end{eqnarray*}
and thus $|IV| \le 4n^\varphi\sqrt{c_0\ept}T^{-\beta/2}\log\,T|\cC_\eta(\{g_t\}_{t=s}^e)|$.
Plugging in the above bounds to the condition given in Lemma \ref{lem:four}, we obtain
\begin{align*}
& |\cC_\eta(\{g_{0, t}\}_{t=s}^e)||\cC_\eta(\{g_{0, t}\}_{t=s}^e)-\cC_{\heta}(\{g_{0, t}\}_{t=s}^e)| >
\\
& C_8n^\varphi\log\,T
\left\{ (\ept T^{-\beta}|\cC_{\eta}(\{g_{0, t}\}_{t=s}^e)|)
\vee (\sqrt{\ept}T^{-\beta/2} |\cC_\eta(\{g_{0, t}\}_{t=s}^e)|)
\vee (n^{\varphi}\log\,T) \right\}
\end{align*}
for some $C_8>0$. 
Due to (\ref{ineq:delta:bs}),
the above is met with $\ept = n^{2\varphi}\uDel_\varphi^{-2}T^{5(1-\beta)}\log^2\,T$.
\end{proof}

With the above lemmas, we are now ready to prove the theorem.
At the beginning of the DCBS algorithm, we have $s=1$ and $e=T$ for which both (\ref{lem:cond:one}) and (\ref{lem:cond:two}) hold,
and thus $\eta_{r_1}$ is estimated by $\heta_{r_1}$ within the distance of $\ept$ from a true change-point (Lemma \ref{lem:two:three}).
Then both of the two segments defined to the left and to the right of $\heta_{r_1}$ satisfy (\ref{lem:cond:one})--(\ref{lem:cond:two})
and so do all the subsequently defined segments, 
and therefore the same arguments apply to show the consistency of the estimates $\heta_{r_2}, \heta_{r_3}, \ldots, \heta_{r_N}$.

Once all the $N$ change-points are detected,
any segment $[s, e]$ determined by two adjacent $\heta_1, \ldots, \heta_N$ (including $\eta_0=1$ and $\eta_{N+1}=T$) satisfy either
\begin{itemize}
\item[(i)] $\exists\, 1 \le r \le N$ such that $r=q_1+1=q_2$ in (\ref{lem:cond:zero}) and
$(\eta_r-s+1) \wedge (e-\eta_r) \le c_6\ept$, or
\item[(ii)] $\exists\, 1 \le r \le N-1$ such that $r=q_1+1$ and $r+1=q_2$, and
$(\eta_r-s+1) \vee (e-\eta_{r+1}) \le c_6\ept$
\end{itemize}
for some positive constant $c_6$.
Under (i), let 
$$(\heta, \wh{m}) = \arg\max_{b\in[s, e]\setminus\cI_{s, e} \, 1 \le m \le n}\cD^\varphi_m(\{|\cX^{(j)}_{s, b, e}|\}_{j=1}^n),$$
and we adopt the notations $y_{0, t}$, $g_{0, y}$ and $\xi_{0, t}$ introduced at the beginning of this section with respect to $\heta$ and $\wh{m}$. Then,
\begin{eqnarray*}
|\cC_{\heta}(\{y_{0, t}\}_{t=s}^e)| &\le& |\cC_{\heta}(\{g_{0, t}\}_{t=s}^e)| + c_4n^\varphi\log\,T \le 
|\cC_{\eta_r}(\{g_{0, t}\}_{t=s}^e)| + c_4n^\varphi\log\,T
\\
&\le& 4\bar{f}n^\varphi(c_6\ept)^{1/2} + c_4n^\varphi\log\,T < \pi^\varphi_{n, T},
\end{eqnarray*}
from the fact that $|g_{0, t}| \le 2\bar{f}n^\varphi$.
A similarly conclusion can be drawn in the case of (ii) as well, and thus the proof is completed. 
\hfill $\square$

\subsection{Proof of Theorem \ref{thm:four}}
\label{sec:pf:thm:four}

We first prove (a) of Theorem \ref{thm:four}, which states that $\cDzh_m$ achieves consistency 
in identifying and locating the change-point under the single change-point scenario.
As in Section \ref{sec:pf:thm:one}, let 
$(\heta_1, \wh{m}_1) = \arg\max_{b\in[1, T]\setminus\cI_{1, T}, 1 \le m \le n} \cDzh_m(\{|\cX^{(j)}_{1, b, T}|\}_{j=1}^n)$,
and adopt the notations $\{k^0_1, \ldots, k^0_n\}$ and $i^0_j\in\{-1, 1\}$ therein.
We define an additive model $y_{0, t} = g_{0, t} + \xi_{0, t}$, where
\begin{eqnarray*}
y_{0, t} = \left\{\gamma_n+\sqrt{\frac{\wh{m}_1(2n-\wh{m}_1)}{2n}}\right\}
\left\{\frac{1}{\wh{m}_1}\sum_{j=1}^{\wh{m}_1}i^0_j \cdot \wt{x}_{k^0_j, t} -
\frac{1}{2n-\wh{m}_1}\sum_{j=\wh{m}_1+1}^{n}i^0_j \cdot \wt{x}_{k^0_j, t}\right\},
\end{eqnarray*}
and $g_{0, t}$ and $\xi_{0, t}$ are constructed analogously with $\wt{f}_{k^0_j, t}$ and $\wt\vep_{k^0_j, t}$ in place of $\wt{x}_{k^0_j, t}$, respectively.

Provided that $n^{-1/2}\gamma_n \to 0$ with $n \to \infty$, 
we have (\ref{xi:bound:one})--(\ref{xi:bound:two}) satisfied with
$\{\xi_{0, t}\}$ replacing $\{\xi_{r, t}\}$ with 
the bound of the order $n^{1/2}\log\,T$.
Also as in (\ref{ineq:delta}), it is shown that
$|\cC_{\heta_1}(\{g_{0, t}\}_{t=1}^T)| \ge
C_9(\gamma_n \vee m_1^{1/2})\wt{\delta}_1\sqrt{\frac{\eta_1(T-\eta_1)}{T}}$
for some constant $C_9>0$, and thus
$\cDzh_{\wh{m}_1}(\{|\cX^{(j)}_{1, \heta_1, T}|\}_{j=1}^n) > \pizh_{n, T}$ with probability tending to one,
which indicates the consistency of the test $\cTzh_{1, T} > \pizh_{n, T}$.
Further, the same arguments adopted in bounding $|\heta_1-\eta_1|$ for Theorem \ref{thm:one}
are applicable to show that $\ept = n\{(\gamma_n \vee m_1^{1/2})\wt{\delta}_1\}^{-2}\log^2\,T$.
Similarly, (\ref{ineq:delta:bs}) can be modified for $\cDzh_m$
and thus the proof of (b) follows verbatim. \hfill $\square$

\begin{supplement}[id=suppA]
  \stitle{Supplement to ``Change-point detection in panel data via double CUSUM statistic''}
  \slink[doi]{}
  \sdatatype{.pdf}
  \sdescription{We provide the detailed description of the Local Bootstrap and proof of some auxillary results.
  In addition, the tables and plots summarising the outcome
  of the simulation studies conducted in Section 5 are presented.}
\end{supplement}

\bibliographystyle{imsart-nameyear}
\bibliography{fbib}

\end{document}